\documentclass[lettersize,journal]{IEEEtran}
\usepackage{amsmath,amsfonts}
\usepackage{ragged2e}
\usepackage{siunitx}
\usepackage{caption}
\usepackage{amsmath}
\usepackage{amssymb}
\usepackage{array}
\usepackage{textcomp}
\usepackage{stfloats}
\usepackage{url}
\usepackage{epsfig}
\usepackage{adjustbox}
\usepackage{verbatim}
\usepackage{graphicx}
\usepackage{cite}
\usepackage{tikz}
\usepackage[flushleft]{threeparttable}
\usepackage{amsthm}
\usepackage{subcaption}
\usepackage[normalem]{ulem}
\usepackage{etoolbox}
\usepackage{amssymb}
\newtheorem{remark}{Remark}

\usepackage{mathtools}

\usepackage{breqn}
\usepackage[multiple]{footmisc}

\definecolor{ForestGreen}{RGB}{34,139,34}
\usepackage{bbm}

\newtheorem{theorem}{Theorem}

\DeclareMathOperator*{\argmax}{arg\,max}

\allowdisplaybreaks

\usepackage{algorithm}
\usepackage{algpseudocode}
\usepackage{setspace}
\usepackage[font=footnotesize,labelfont=bf,skip=5pt]{caption}
\usepackage{subcaption}
\usepackage{enumitem}
\usepackage{bigints}
\usepackage[colorlinks]{hyperref}
\usepackage{flushend}
\hypersetup{
	urlcolor     = magenta, 
	linkcolor    = blue, 
	filecolor	    = cyan, 
	citecolor   = {blue}, 
}

\usepackage{scalerel}[2016-12-29]

\usepackage{wrapfig}
\usepackage[outdir=./]{epstopdf}
\usepackage{multirow}
\begin{document}
\title{Practical RIS Gain without the Pain via Randomization and Opportunistic Scheduling in 5G~NR Wireless Systems: Theory and Experiments}
 \author{Nekkanti Guna Sai Kiran, L. Yashvanth,~\emph{Member, IEEE}, Raju Malleboina,~\emph{Student Member, IEEE}, Venkatareddy Akumalla,~\emph{Member, IEEE},  Debdeep Sarkar,~\emph{Senior Member, IEEE}, and Chandra R. Murthy,~\emph{Fellow, IEEE}
 \thanks{N. G. S. Kiran, R. Malleboina, V. Akumalla, D. Sarkar, and C. R. Murthy are with the Dept. of ECE, Indian Institute of Science (IISc), Bengaluru, India 560012. E-mails: \{gunanekkanti, malleboinar\}@iisc.ac.in, venkatareddy@fsid-iisc.in, \{debdeep, cmurthy\}@iisc.ac.in.} \thanks{L. Yashvanth is with Dept. of EEE, University College London, WC1E 6BT London, U.K. He was with Dept. of ECE, IISc, India during the development of this work. E-mail: l.yashvanth@ucl.ac.uk, yashvanthl@iisc.ac.in.}
 \thanks{The authors thank
Prof. K. J. Vinoy, Dept. of ECE, IISc, Bengaluru for providing directional horn
antennas for the experiments.}
\thanks{A part of this work has been accepted at the IEEE International Conference on Communications (ICC), 2026, Glasgow, U.K. \cite{ICC_Paper}}
}
 
 \maketitle
\begin{abstract}
In this paper, we theoretically analyze and experimentally demonstrate the performance gains achievable by integrating an in-house built reconfigurable intelligent surface (RIS) with a 5G new radio (NR) system implemented using the OpenAirInterface (OAI) software stack. 
Unlike conventional RIS-assisted systems that rely on explicit channel state information (CSI) estimation followed by RIS phase configuration optimization, we adopt a low-complexity approach in which the RIS phase states are randomly switched among predefined configurations. The resulting channel fluctuations are opportunistically exploited by the inherent proportional fair (PF) scheduling mechanism of 5G NR. We develop a theoretical framework that characterizes the interaction between RIS switching dynamics and PF scheduling. Based on this framework and the associated analysis, we provide design guidelines for selecting the RIS switching time $T_s$ and the PF throughput averaging window $T_c$ that maximize the system throughput. Experimental evaluations on the 5G NR testbed demonstrate improvements in key performance metrics, including reference signal received power (RSRP), block error rate (BLER), modulation and coding scheme (MCS) index, and throughput. Our key takeaway is that randomly configured RIS operation with appropriately chosen system parameters can achieve performance comparable to optimized RIS designs, with no additional overhead compared to a conventional 5G NR system. More importantly, it requires no coordination between the RIS and the 5G NR system.
\end{abstract}
\begin{IEEEkeywords}
Reconfigurable intelligent surfaces, 5G NR, OpenAirInterface, Proportional-fair scheduling, Experimental validation.
\end{IEEEkeywords}
\color{black}

\section{Introduction}
Reconfigurable intelligent surfaces (RISs) enable one to control the wireless propagation environment by adjusting the phase shifts of many low-cost passive reflecting elements \cite{Basar_IA_2019}. It improves signal strength, coverage, and throughput at very low hardware and energy cost, making it a promising technology for beyond-5G systems \cite{Renzo_EURASIP_2019,Emil_arxiv_2021}. Most RIS studies in the literature rely on accurate channel state information (CSI) of the cascaded base station (BS)-RIS-user links and optimized phase shift configuration in order to enhance the signal-to-noise ratio (SNR) and throughput \cite{Wu_GLOBECOM_2018}. However, achieving these gains in practical wireless systems involves significant overhead in channel estimation, computational complexity, and control signaling~\cite{Mishra_ICASSP_2019,Nadeem_JCS_2020, Cai_2022_TCOM, Wu_ICOMM_2020}. 
These challenges are exacerbated in dynamic multiuser environments, where rapidly varying channels require frequent channel re-estimation and phase re-optimization. Further, integrating RISs into the 5G new radio (NR) protocol, accounting for its CSI reporting, modulation and coding scheme (MCS) selection, and medium access control (MAC) scheduling mechanisms, can significantly reduce the gains of optimization-based designs.  
Hence, there is a need for RIS strategies that exploit its reflectivity properties without relying on explicit CSI acquisition and phase optimization. 

In this work, we study schemes where the RIS phases are randomly sampled independent of the instantaneous CSI. The key idea is that in multi-user systems, random RIS configurations create large channel variations which can then be exploited through opportunistic scheduling~\cite{Viswanath_TIT_2002}. In particular, for a given randomly chosen RIS configuration, with high probability, at least one user equipment (UE) will experience favorable channel conditions; opportunistically scheduling that UE procures RIS benefits without incurring the overheads mentioned above.  
Although prior works on such ideas exist~\cite{Nadeem_WCL_2021,Yashvanth_TSP_2023,Yashvanth_TSP_2026}, they assume a large number of UEs and do not account for practical trade-offs that emerge during the implementation of these schemes.

Most of the existing studies on RIS-assisted communication systems focus on theoretical analysis and simulations, while experimental validations of RIS-assisted systems remain under-explored. Some of the early experimental works proposed amplitude-and-phase modulation techniques and low-resolution RIS prototypes \cite{Tang_JSAC_2020,Linglong_Exp_Access_2020}, and subsequent efforts developed realistic reflection models \cite{Lin_Exp_TCOM_2024}, examined RIS deployment strategies \cite{Yanqing_Exp_WCL_2023}, and demonstrated indoor and outdoor performance gains using optimized RIS configurations \cite{Emil_Exp_TCOM_2021}. Other experimental studies include analysis of RIS-assisted path loss models \cite{Wankai_Exp_TWC_2021}, understanding the performance impact of RISs in multiple mobile operator systems~\cite{lodro2022experimental}, demonstration of coverage extension \cite{Kayraklik_Exp_ICC_2023}, comparison of active and passive RIS architectures \cite{Linglong_Exp_TCOM_2023}, and bit error rate and image communication~\cite{Arun_2025}.
However, all these studies rely on vector network analyzers, signal generators, or software-defined radios with 5G system simulators or emulators, and do not consider end-to-end over-the-air (OTA) 5G NR systems. Consequently, the impact of practical aspects such as CSI acquisition overhead, MAC-layer scheduling behavior, MCS adaptation, and retransmission mechanisms on RIS-assisted performance remains largely unexplored. 

To address these gaps, we leverage the OpenAirInterface (OAI) 5G NR protocol stack, which enables the development and testing of OTA systems compliant with 3GPP standards~\cite{OAI_7,openairinterface_testbed}. 
Although ~\cite{sahin2025ris} considers the performance of RIS in real-time 5G setups, it is based on CSI-dependent optimized phase configuration, which incurs significant complexity. 
In our setup, the RIS phase coefficients are randomly sampled from a predefined distribution and is reconfigured at an appropriate switching interval, which introduces fluctuations in the composite wireless channel. Then, using a proportional fair (PF) scheduler with an appropriately chosen exponentially weighted moving average (EWMA) constant, the 5G BS exploits the channel fluctuations and opportunistically schedules a UE with the favorable channel conditions for data transmission. Our key contributions are as follows:
\begin{enumerate}
    \item We develop an end-to-end 5G NR testbed integrating the OAI framework and an in-house-built RIS. Then, we $a)$ demonstrate the end-to-end working mechanism of a real-time RIS-assisted 5G NR system, and $b$) quantify the performance gains offered by the RIS. 
    \item Under finite switching time of the RIS states, we analytically prove that, when $T_c$, the PF scheduler's EWMA constant, becomes large, the performance of  randomized RIS with phases sampled from an appropriate distribution converges to that obtained by optimized RIS phase under round-robin (RR) scheduling of UEs. (See Theorem~\ref{thm:infinite_Tc}.)
    \item We analytically derive the relationship between $T_c$ and the key system parameters such as the number of UEs, the RIS switching interval, and the number of RIS states, such that the achieved throughput is arbitrarily close to its optimal value. (See Theorem~\ref{thm:finite_Tc}.) 
    \item  We study the temporal variation of the reference signal received power (RSRP), MCS index, block error rate (BLER), and throughput, and experimentally demonstrate that a randomly tuned RIS offers a performance close to an optimized RIS, under PF scheduling.
    \item Finally, we show that even when the UE locations are arbitrary, the performance in the presence of a randomized RIS is always significantly better than a system without an RIS. Thus, even na\"ive randomized RIS-designs can be beneficial in real-time 5G wireless systems.
\end{enumerate}
The key takeaway from this work is that integrating an RIS into 5G systems need not introduce substantial complexity. Randomized RIS, along with native 5G NR scheduling mechanisms, are sufficient to deliver most of the gains commonly attributed to optimization-based RIS designs, offering a practical and scalable path for real-world RIS deployment.
\section{System Description}
\subsection{Mathematical Model}\label{Mathematical_discription}
We consider a system in which a BS, referred to as a gNodeB (gNB) in 5G NR, serves $K$ UEs with the assistance of an RIS comprising $N$ reflecting elements. The RIS enhances the BS–UE links, particularly in scenarios where the direct line-of-sight (LoS) paths are blocked. 
Let $\mathbf{h}_1 \in \mathbb{C}^{N}$ denote the channel from the BS to the RIS, and let $\mathbf{h}_{2,k} \in \mathbb{C}^{N}$ be the channel from the RIS to UE-$k$. During time slot $t$, the received signal at UE-$k$ can be expressed as
\begin{align}
   y_k(t) = \sqrt{\beta_{k}} \mathbf{h}_{2,k}^{T} \boldsymbol{\Phi}(t) \mathbf{h}_1 x_k(t) + n_k(t), \label{eq:yk2}
\end{align}
where $\beta_{k}$ denotes the large-scale path loss of the cascaded BS-RIS-UE-$k$ channel, and $\boldsymbol{\Phi}(t)$ is the diagonal RIS phase-shift matrix given by 
$\boldsymbol{\Phi}(t) = \mathrm{diag}\!\left(e^{j\theta_1(t)}, e^{j\theta_2(t)}, \dots, e^{j\theta_N(t)}\right)$ with $\theta_n(t)$ being the phase shift introduced at $n$th RIS element. Further, $x_k(t)$ is the $k$th UE's transmit symbol, and $n_k(t)$ represents the additive Gaussian noise at UE-$k$.
For notational convenience, we define the cascaded BS-RIS-UE-$k$ channel as 
$\mathbf{h}_{c,k} \triangleq \mathbf{h}_1 \odot \mathbf{h}_{2,k}$, where $\odot$ denotes the Hadamard (element-wise) product. 
Furthermore, let the RIS phase vector (with conjugation) be defined as 
$\boldsymbol{\phi}(t) \triangleq \left[e^{-j\theta_1(t)}, \dots, e^{-j\theta_N(t)}\right]^T$.
Then, we can rewrite~\eqref{eq:yk2} as
\begin{equation}
    y_k(t) = \sqrt{\beta_{k}} \boldsymbol{\phi}^{H}(t) \mathbf{h}_{c,k} x_k(t) + n_k(t). \label{eq:yk3}
\end{equation}
Accordingly, the effective end-to-end channel between the BS and UE-$k$ at time slot $t$ is given by 
\begin{equation}
    g_k(t) = \sqrt{\beta_{k}} \boldsymbol{\phi}^{H}(t) \mathbf{h}_{c,k}. \label{eq_overall_ch_UE_k}
\end{equation}
We next describe the features of the in-house-Built RIS and the OAI-based $5$G network used in this work. 
\subsection{Design and Modeling of the In-house-Built RIS}
We have built a one-bit-coded RIS operating at $5$~GHz with a bandwidth of $300$~MHz \cite{Malleboina_EuCAP_2025}. The RIS is implemented as a single-layer $32 \times 32$ uniform planar array (UPA) of $1024$ unit cells, each measuring $\lambda_0/4 \times \lambda_0/4$ (at $5$~GHz, $\lambda_0 = 60$~mm). This compact design minimizes reflection loss while ensuring mechanical stability and ease of fabrication. Each unit cell integrates an SMP1302--040LF PIN diode to switch between two reflection states, on and off, with corresponding reflection phases of $-28^\circ$ and $150^\circ$, yielding a phase difference close to $180^\circ$ for efficient one-bit digital coding. Beam steering is achieved by analytically calculating the reflection-phase gradient across the RIS surface \cite{Raju_AWPL_2023}, then quantizing it into discrete $0^\circ$ and $180^\circ$ phase levels. Full-wave CST simulations confirm scan-loss-free steering from $-60^\circ$ to $60^\circ$, with a half-power beamwidth of $6^\circ$ and sidelobe levels below $-9$~dB. 

The RIS is controlled by an field-programmable gate array (FPGA)-based biasing module, enabling real-time reconfiguration and adaptive beam steering. DC isolation and biasing lines are carefully included to ensure stable operation. Finally, anechoic chamber experiments validate the design, showing measured radiation patterns consistent with simulations, with accurate beam direction and low sidelobe distortion~\cite{Malleboina_EuCAP_2025}.  Thus, the RIS achieves wide-angle beam steering with superior angular resolution compared to conventional reflectarrays. The analytical phase-distribution method simplifies codebook generation and reduces computational load, while the use of PIN diodes and radial stubs lowers cost, reduces biasing complexity, and scales efficiently to large arrays. 
\subsection{RIS-Aided 5G NR via OAI framework}\label{sec_OAI_description}
\subsubsection{Testbed Architecture}\label{sec_testbed_architecture}
The $5$G NR testbed is implemented using the OAI platform \cite{oai_platform2014}, in a monolithic stand-alone (SA) configuration that supports the $5$G NR radio access protocol (RAN) protocol stack and core network. In our work, universal software radio peripherals (USRPs) B$210$ serve as radio units (RUs) to enable OTA communication between the gNB and the UEs. The OAI $5$G RAN and UE software stacks run on an Intel i$7$ desktop and laptops equipped with $12$~cores running at $4.9$~GHz, $32$~GB RAM, and Ubuntu $22.04$~LTS with low-latency kernels. For synchronization, an OctoClock provides a $10$~MHz reference 
signal to all devices.
\subsubsection{Communication Resources}
The system operates over a $40$~MHz bandwidth centered at $\sim 5$~GHz. With numerology set to~$1$ ($30$~kHz subcarrier spacing), each time slot has a duration of $0.5$~ms and contains $14$ orthogonal frequency division multiplexing (OFDM) symbols. A time-division duplexing (TDD) scheme with a downlink (DL) – uplink (UL) periodicity of $10$ slots ($5$~ms) is employed, With $6$ slots assigned to DL, $3$ to UL, and $1$ slot is configured as a mixed slot. 
Across the $40$~MHz band, the system uses $106$ physical resource blocks (PRBs), with up to $13$ OFDM symbols per PRB being allocated for DL transport blocks (TB) in each time slot~\cite{3gpp_NR_MCS_table}. 
\subsubsection{Discrete Rate-Adaptation and Scheduling}\label{sec_MCS}
OAI uses discrete rate-adaptation to set the DL throughput according to instantaneous channel conditions. The adaptation is governed by the MCS, which specifies (a) the modulation order and (b) the code rate of a forward error correction (FEC) code. In the 3GPP-defined $64$-QAM MCS Index Table for the physical DL shared channel (PDSCH)~\cite{3gpp_NR_MCS_table}, there are 29 possible MCS indices. Higher indices correspond to higher-order constellations/rates, and yield larger spectral efficiency (SE). In OAI, the MCS index selection primarily depends on the BLER, and can vary from $3$ to $27$. The BLER is estimated using the hybrid automatic repeat request (HARQ) mechanism, which handles TB retransmissions~\cite{Lagen_Exp_ICC_2020}, by counting TB failures over a $100$~ms window. The MCS index is increased or deceased by $1$ when the BLER falls below $0.05$ or exceeds $0.15$, respectively, to approximately maintain the BLER at a target of $0.1$. If a negative acknowledgment (NACK) is received, the same TB is retransmitted up to $3$ times using the same MCS before being discarded. User scheduling in OAI is performed by a PF scheduler, which implements an \emph{opportunistic} UE selection rule to maximize the system throughout while maintaining fairness. At time slot $t$, the BS schedules a UE with index $k^*(t)$ across all PRBs according to
\begin{equation}\label{eq:PF_metric}
k^*(t) = \argmax\nolimits_{k \in \{1,2,\ldots,K\}} \ \ \eta_{k,\textrm{PF}}(t) \triangleq {R_k(t)}\Big/{T_k(t)},
\end{equation}
where $\eta_{k,\textrm{PF}}(t)$ is the so-called \emph{PF metric} of UE-$k$ and $R_k(t) = \log_2\left(1+\left|g_k(t)\right|^2{P}\Big/{\sigma^2}\right)$ denotes the instantaneous SE of UE-$k$, 
where $g_k(t)$ is the effective channel at UE-$k$ given by \eqref{eq_overall_ch_UE_k}, $P$ is the BS transmit power, and $\sigma^2$ is the noise variance, and $T_k(t)$ is the average SE of UE-$k$ till time $t$ 
updated using an EWMA, given by
\begin{equation}\label{eq_PF_scheduler_Tk_update}
T_k(t+1) = \left(1-\tfrac{1}{T_c}\right)T_k(t) + \mathbbm{1}\!\left\{k = k^*(t)\right\} \frac{1}{T_c} R_k(t).
\end{equation}
In~\eqref{eq_PF_scheduler_Tk_update}, $T_c > 1$ is the EWMA constant, also called the PF window length. A larger $T_c$ yields larger long-term throughput by favoring more opportunism in UE scheduling, while a smaller $T_c$ prioritizes short-term fairness across UEs.
\subsection{Overall Mechanism of RIS-aided 5G NR}
As shown in  Fig.~\ref{fig:FULL_BLOCK_DIAGRAM}, the system consists of a gNB, multiple UEs, and an FPGA-controlled RIS. The 
FPGA switches the RIS phase vector at intervals of $T_s$; thus, the RIS configuration remains fixed for $T_s$~s before switching to a new phase state. 
Each UE monitors the DL channel quality through standards-mandated procedures and
reports the RSRP to the gNB with a periodicity of approximately $80$~ms. In addition, the UE that is scheduled for data transmission performs cyclic redundancy check (CRC) on the decoded data and feeds back ACK/NACK information to the gNB, indicating successful or unsuccessful decoding, respectively.
Upon receiving an ACK, the gNB schedules the next UE according to the PF metric defined in~\eqref{eq:PF_metric}, using an MCS that is adaptively adjusted according to the observed BLER of that UE, inferred from HARQ feedback. 
However, if a NACK is received, the gNB prioritizes retransmitting the packets to the already scheduled UE using the same MCS under the chase-combining (CC)-based HARQ operation, instead of scheduling a new UE. 

Once the UE is scheduled, resource allocation is performed at the gNB, which includes symbol mapping according to the selected MCS and OFDM modulation over a $40$~MHz bandwidth. The resulting baseband signal is upconverted to the carrier frequency of $5$~GHz and transmitted.
The transmitted signal propagates through the RIS-assisted channel and is received by the scheduled UE, where channel estimation, demodulation, decoding, and CRC verification are performed again, and this process continues. It is important to note that the choice of RIS configuration determines the instantaneous channel conditions and impacts the RSRP reports, selected MCS levels, HARQ ACK/NACK outcomes, and the overall throughput and fairness performance across UEs.
\begin{figure*}[t]
    \centering
\includegraphics[width=0.8\linewidth]{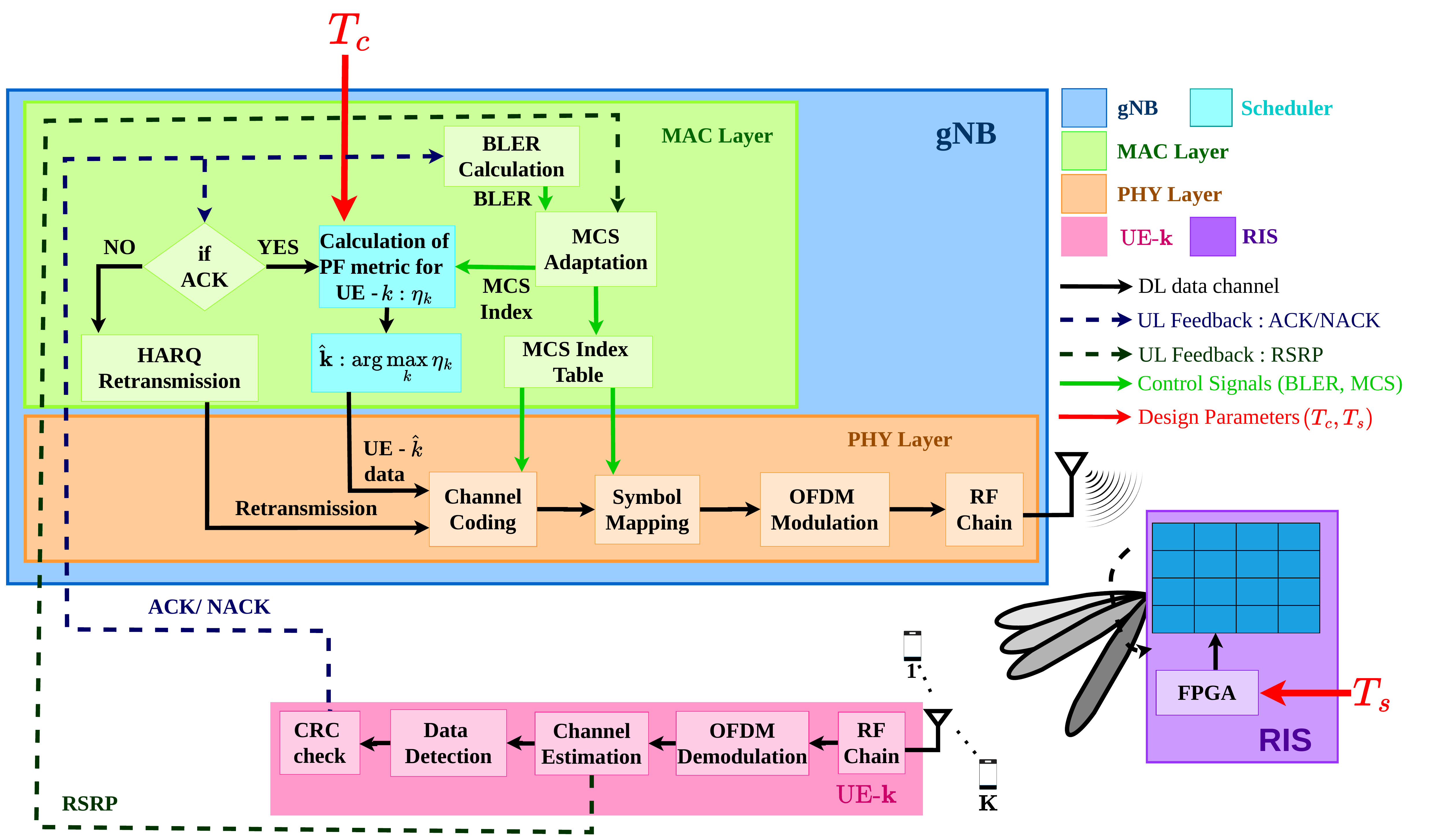}
    \caption{Block-diagram of the overall RIS-assisted $5$G NR-based multi-UE communication system.}
    \label{fig:FULL_BLOCK_DIAGRAM}
\end{figure*}
\subsection{Problem Description}
One of the objectives of this paper is to demonstrate the integration of a RIS into a $5$G NR-compliant system. A second objective is to show that randomly selected RIS phase configurations, when combined with the inherent opportunistic scheduling of $5$G NR, can deliver performance gains \emph{comparable to optimized RIS operation} while avoiding additional system complexity. 
To this end, we experimentally investigate:
\begin{enumerate}[leftmargin=*]
    \item The achievable performance gain in the presence of RIS compared to the system without RIS,
    \item The variation of RSRP, MCS index, BLER, and throughput across UEs under randomly sampled RIS configurations with PF scheduling used at the 5G gNB, 
    \item Comparison of the throughput of randomized RIS under PF scheduling and optimized RIS under RR scheduling,
    \item The relationship between RIS switching time and the EWMA constant of the PF scheduler, and their impact on the performance,
    \item The utility of randomized RIS in terms of the throughput, when the UEs are arbitrarily located in the area of interest. 
\end{enumerate}
Furthermore, we provide a rigorous theoretical analysis to explain the key results observed in the randomized RIS-aided $5$G NR system, as described next. 
\section{Randomized RIS-Assisted Opportunistic Scheduling in $5$G NR: Theory}\label{Theory}
In this section, we theoretically establish that a randomized RIS in conjunction with the native PF scheduling in $5$G NR can procure the optimal benefits of an RIS. We also characterize the system parameter settings that realize the gains in practice. To this end, we first present the \emph{sampling} distribution from which the random RIS phase vectors are drawn in every time slot.
\subsection{Choice of Random RIS Configurations}
\label{Sec_RIS_distribution}
Recall from Sec.~\ref{Mathematical_discription} that the $N$-element RIS is configured
as an $N_x \times N_y$ UPA with $N = N_x N_y$, where
$N_x = N_y = 32$ and the inter-element spacing is $d_x = d_y = d = \lambda_0/4$.
Let $(\nu,\psi)$ denote the tuple corresponding to azimuth and elevation angles at the RIS.
The RIS array response vector is given by
\begin{equation}
\mathbf{a}_{\mathrm{RIS}}(\nu,\psi)
=
\mathbf{a}_x(\nu,\psi) \otimes \mathbf{a}_y(\nu,\psi),
\label{eq:effective_array}
\end{equation}
where the steering vector components are given by 
$[\mathbf{a}_x(\nu,\psi)]_m = e^{j \frac{2\pi}{\lambda} m d_x \cos\psi \cos\nu}$ for $m=0,\ldots,N_x-1$, 
and 
$[\mathbf{a}_y(\nu,\psi)]_n = e^{j \frac{2\pi}{\lambda} n d_y \cos\psi \sin\nu}$ for $n=0,\ldots,N_y-1$.

Here, $\otimes$ denotes the Kronecker product, and $\mathbf{a}_{\mathrm{RIS}}(\nu,\psi) \in \mathbb{C}^{N}$
is the overall array response vector for the UPA-based RIS.
Then, the channels in different links can be modeled as 
$\mathbf{h}_1 = \mathbf{a}_{\mathrm{RIS}}(\nu_{\mathrm{BR}},\psi_{\mathrm{BR}})$ 
and 
$\mathbf{h}_{2,k} = \alpha_k\,\mathbf{a}_{\mathrm{RIS}}(\nu_k,\psi_k)$.
where $\alpha_k
\overset{\text{i.i.d.}}{\sim} \mathcal{CN}(0,1)$ is the fading component in the channel at UE-$k$, and $(\nu_{\mathrm{BR}}, \psi_{\mathrm{BR}})$ and $(\nu_k,\psi_k)$ represent the azimuth-elevation angle pair at the RIS corresponding to the direction of arrival from BS and direction of departure to UE-$k$, respectively. 
Then, the cascaded fading channel becomes 
$\mathbf{h}_{c,k} = \alpha_k \, \mathbf{a}_{\mathrm{RIS}}(\nu_{c,k},\psi_{c,k})$,
where $(\nu_{c,k},\psi_{c,k})$ denotes the \emph{cascaded} azimuth--elevation
angle pair at the RIS satisfying 
\begin{align}
    \cos\psi_{c,k}\cos\nu_{c,k} &=
\cos\psi_{\mathrm{BR}}\cos\nu_{\mathrm{BR}} +
\cos\psi_k\cos\nu_k,  \\
\cos\psi_{c,k}\sin\nu_{c,k} &=
\cos\psi_{\mathrm{BR}}\sin\nu_{\mathrm{BR}} +
\cos\psi_k\sin\nu_k.
\end{align}

Substituting the above in \eqref{eq_overall_ch_UE_k}, the effective channel at UE-$k$ is
\begin{equation}
g_k(t)
=
\sqrt{\beta_k}\,\alpha_k\,
\boldsymbol{\phi}^H(t)\,
\mathbf{a}_{\mathrm{RIS}}(\nu_{c,k},\psi_{c,k}).
\label{eq:gk_upa_az_el}
\end{equation}
Now, following the codebook approach to select beamforming (BF) vectors as in $5$G NR~\cite{3gpp_NR_MCS_table}, let the sample space of RIS phase configuration vectors be
$\{\boldsymbol{\phi}_1,\ldots,\boldsymbol{\phi}_L\}$, and $\boldsymbol{\phi}_\ell \in \mathbb{C}^N$, $\ell\in[L]$
is selected with probability $p_\ell$, with $p_\ell\ge 0$, $\sum_{\ell=1}^L p_\ell=1$. Then,
$\mathcal{S}
\triangleq
\left\{
\left(\boldsymbol{\phi}_1,p_1\right),
\ldots,
\left(\boldsymbol{\phi}_L,p_L\right)
\right\}$
defines the sampling distribution of RIS configurations. Using \cite[Theorem~$1$]{Yashvanth_WCL_2025}, the optimal conditional probability density function (PDF) of the RIS vector given a cascaded channel is 
$g_{\boldsymbol{\theta}}^{\mathrm{opt}}(\boldsymbol{\theta}')
=
\delta\!\left(\boldsymbol{\theta}' - \mathcal{F}(\mathbf{h}_{c,k})\right)$,
where $\delta(\cdot)$ denotes the Dirac delta function and $\mathcal{F}(\cdot)$ represents the \emph{phase-mapping rule}:
\begin{equation}
    \mathcal{F}:\mathbb{C}^{N}\rightarrow\mathcal{C}: \mathbf{h}_{c,k}\mapsto \boldsymbol{\theta}'=\exp\!\left(j\,\angle \mathbf{h}_{c,k}\right),
\end{equation}
and $
\mathcal{C}
\triangleq
\left\{
\boldsymbol{\theta} \in \mathbb{C}^{N}
\;\middle|\;
|\theta_n| = 1,\;
n=1,\ldots,N
\right\}.$
In other words $\mathcal{F}(\cdot)$ maps $\mathbf{h}_{c,k}$ to its
corresponding RIS phase configuration with unit-modulus entries. 
Then, by the law of total probability with respect to the channel distribution,
the unconditional PDF of the RIS phase configuration vector is given by
\begin{equation}
f_{\boldsymbol{\theta}}^{\mathrm{opt}}(\boldsymbol{\theta}')
= \int
\delta\big(\boldsymbol{\theta}' - \mathcal{F}(\mathbf{h})\big)
p_{\mathbf{h}_c}(\mathbf{h})\, d\mathbf{h},
\label{eq:RIS_density_general}
\end{equation}
where $p_{\mathbf{h}_c}(\cdot)$ is the PDF of the cascaded channel vector. 

We tile the cascaded azimuth-elevation domain into $L$
regions $\{\mathcal{A}_\ell\}_{\ell=1}^L$ such that $\bigcup_{\ell=1}^L \mathcal{A}_\ell = \mathcal{D}$, and $\mathcal{A}_\ell \cap \mathcal{A}_m = \varnothing,\ \ell \neq m$, 
where $\mathcal{D}$ denotes the range of the cascaded angles across UEs. Then, for each region $\mathcal{A}_\ell$, $\ell=1,\ldots,L$, we define a representative cascaded azimuth-elevation angle pair $(\nu^{(\ell)},\psi^{(\ell)})$, which is chosen as the centroid of $\mathcal{A}_\ell$, given by 
\begin{equation}
(\nu^{(\ell)},\psi^{(\ell)}) \triangleq 
\frac{1}{|\mathcal{A}_\ell|}
\iint_{\mathcal{A}_\ell}
(\nu_c,\psi_c)\, d\nu_c\, d\psi_c,
\end{equation}
where $|\mathcal{A}_\ell| \triangleq \iint_{\mathcal{A}_\ell} d\nu_c\, d\psi_c$ denotes the Lebesgue measure of $\mathcal{A}_\ell$.
Then, we choose the $\ell$th RIS phase vector as 
$\boldsymbol{\phi}_\ell = \mathbf{a}_{\mathrm{RIS}}(\nu^{(\ell)},\psi^{(\ell)}) \in \mathcal{C}$.

Next, we characterize $f_{\nu_c,\psi_c}(\nu_c,\psi_c)$, the probability mass function (PMF) of the RIS phase vectors. Let $\{(\nu_{c,k}, \psi_{c,k})\}_{k=1}^{K}$ denote the cascaded azimuth-elevation angle pairs corresponding to the $K$ UEs. When the direct BS-UE link is blocked, the throughput-optimal random RIS phase configurations depend only on the cascaded angles across UEs~\cite[Sec~$4.2$]{Yashvanth_ICASSP_2023}. Accordingly,
\begin{equation}
    f_{\nu_c,\psi_c}(\nu_c,\psi_c)
=
\frac{1}{K}\sum\nolimits_{k=1}^{K}
\delta\big(\nu_c-\nu_{c,k}\big)\,
\delta\big(\psi_c-\psi_{c,k}\big).
\end{equation}
Using~\eqref{eq:RIS_density_general}, the PMF corresponding to the RIS vector $\boldsymbol{\phi}_\ell$ is
\begin{equation}
p_\ell =
\iint_{\mathcal{A}_\ell}
f_{\nu_c,\psi_c}(\nu_c,\psi_c)\, d\nu_c\, d\psi_c
=
\frac{1}{K}\sum_{k=1}^{K}
\mathbbm{1}_{\{(\nu_{c,k},\psi_{c,k}) \in \mathcal{A}_\ell\}}.
\end{equation}
 In other words, $p_\ell$ equals the fraction of UEs whose cascaded angles lie in the region $\mathcal{A}_\ell$.
Therefore, the sampling distribution $\mathcal{S}$ for the RIS phase configurations can be written as
\begin{multline}\label{eq_final_sampling_distribution}
\mathcal{S}
=
\Big\{
\big(\boldsymbol{\phi}_\ell, p_\ell\big)
\,\Big|\,
\boldsymbol{\phi}_\ell
=
\mathbf{a}_{\mathrm{RIS}}(\nu^{(\ell)},\psi^{(\ell)}), \\
p_\ell
=
\frac{1}{K}
\sum_{k=1}^{K}
\mathbbm{1}\big\{
(\nu_{c,k},\psi_{c,k}) \in \mathcal{A}_\ell
\big\},
\;
\ell = 1,\ldots,L
\Big\}.
\end{multline}
\begin{figure}[t]
    \centering
    \includegraphics[width=\linewidth]{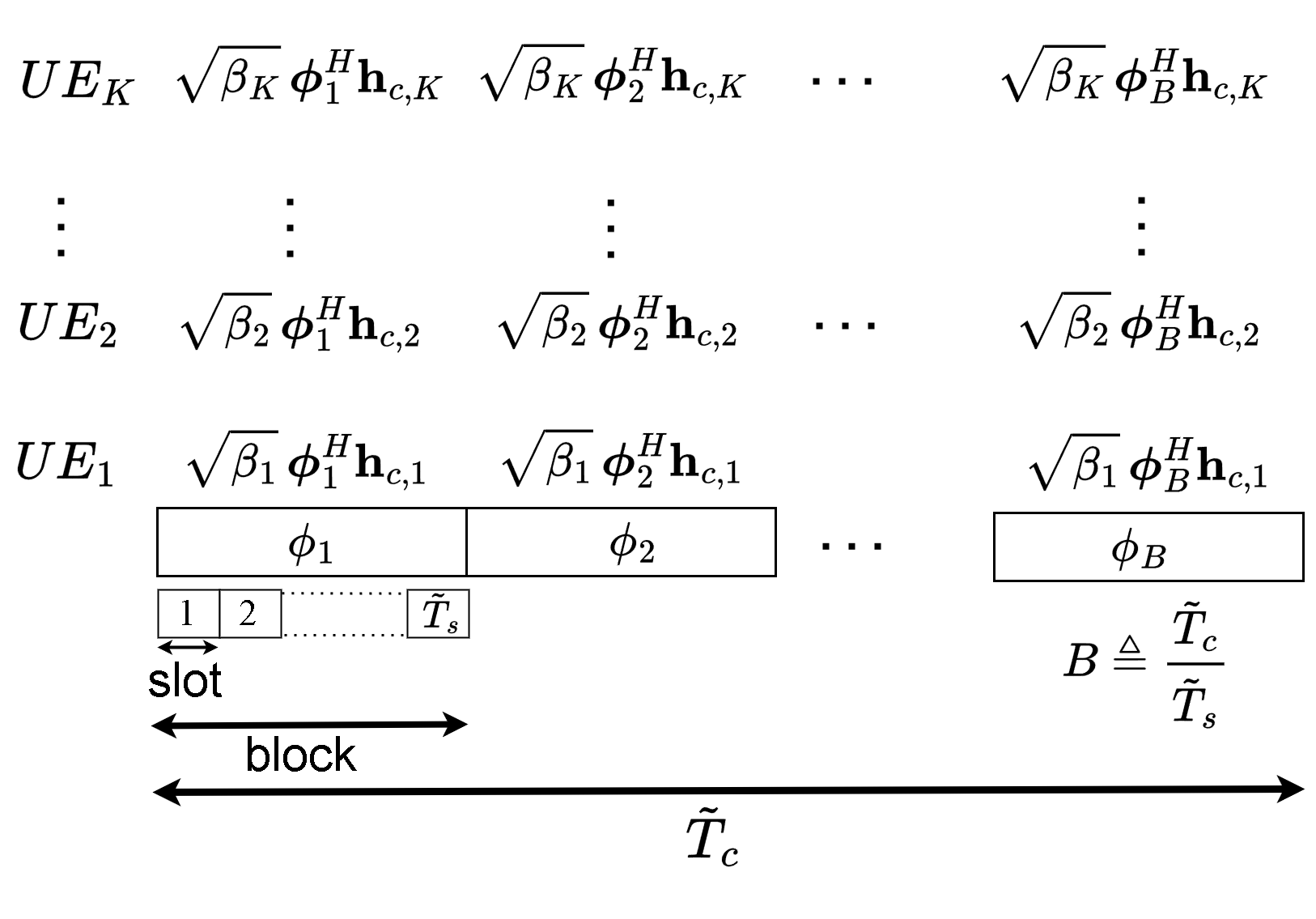}
    \caption{Timing structure. The RIS vector $\boldsymbol{\phi}_\ell$ remains fixed for $\tilde{T}_s$ slots, and independently draws a new random state in the subsequent $\tilde{T}_s$ slots.}
    \label{fig:scheduling_blocks}
\end{figure}
\subsection{Analysis of Randomized RIS Opportunistic Scheduling}
The RIS updates its configuration every $\tilde{T}_s \triangleq \frac{T_s}{\Delta t}$ time slots, where $\Delta t$ is the slot duration. In each block of $\tilde{T}_c$ slots, the RIS adopts a fixed random phase vector drawn from the sampling distribution $\mathcal{S}$ in~\eqref{eq_final_sampling_distribution}; see Fig.~\ref{fig:scheduling_blocks} for an illustration. Thus, during the $n$th block, the RIS configuration is
\begin{equation}
    \boldsymbol{\phi}(t) = \boldsymbol{\phi}_\ell, \qquad t \in [n\tilde{T}_s, (n+1)\tilde{T}_s - 1],
\end{equation}
where $\boldsymbol{\phi}_\ell$ is selected with probability $p_\ell$, $\boldsymbol{\phi}(t)$ denotes the RIS configuration at time slot $t$, and $n$ is the block index. Thus, $\tilde{T}_s = 1$ corresponds to per-slot RIS reconfiguration, whereas $\tilde{T}_s = \infty$ represents a fixed RIS configuration over all time.
\textcolor{black}{Now, the instantaneous throughput of UE--$k$ at time $t$, denoted by $R_k^{\tilde{\ell}}(t)$ when the RIS vector assumes $\boldsymbol{\phi}_{\tilde{\ell}}$ is given by
\begin{equation}
\label{eq:Rk_inst}
R_k^{\tilde{\ell}}(t)
=
\log_2\!\left(
1+
\frac{P}{\sigma^2}\,
\beta_k |\alpha_k|^2
\left|
{\boldsymbol{\phi}_{\tilde{\ell}}}^{H}
\mathbf{a}_{\mathrm{RIS}}(\nu_{c,k},\psi_{c,k})
\right|^2
\right),
\end{equation}
where the effective cascaded channel follows from \eqref{eq:gk_upa_az_el}.}
Let $R_k^{\mathrm{opt}}$ denote the \emph{optimal} throughput achieved when the RIS vector is optimized to maximize the throughput at UE-$k$, i.e.,
\begin{equation}\label{eq_opt_benchmark}
R_k^{\mathrm{opt}}
=
\log_2\!\left(
1+
\frac{P}{\sigma^2}\,
\beta_k |\alpha_k|^2
\left|
{\boldsymbol{\phi}^{\mathrm{opt}}_k}^{H}
\mathbf{a}(\nu_{c,k},\psi_{c,k})
\right|^2
\right),
\end{equation}
where $\boldsymbol{\phi}^{\mathrm{opt}}_k \!\!=\!\! \argmax\nolimits_{\boldsymbol{\phi}\in \{\boldsymbol{\phi}_1,\ldots,\boldsymbol{\phi}_L\}} \left|{\boldsymbol{\phi}}^H\mathbf{a}(\nu_{c,k},\psi_{c,k})\right|^2$\!.
We have the following result. 
\begin{theorem}\label{thm:infinite_Tc}
Consider an RIS-assisted $5$G system with $K$ UEs and $L$ RIS configuration states, where, in each block comprising $T_s \in (0,\infty)$~seconds, the RIS state is independently and randomly sampled according to~\eqref{eq_final_sampling_distribution}. The long-term DL system throughput achieved by a PF scheduler with EWMA window size $T_c >1 $, denoted by $R_{\mathrm{sys},T_c}^{(K)}$, satisfies
\begin{equation}\label{eq_thm_statement}
\lim_{T_c \to \infty}
R_{\mathrm{sys},T_c}^{(K)}
=
(1/K)\sum\nolimits_{k=1}^{K} R_k^{\mathrm{opt}},
\end{equation}
where $R_k^{\mathrm{opt}}$ is the optimal throughput at UE-$k$ as in~\eqref{eq_opt_benchmark}.
\end{theorem}

\begin{proof}
See Appendix~\ref{appendix_proof_thm1}.
\end{proof}
We interpret Theorem~\ref{thm:infinite_Tc} as follows. As $T_c \to \infty$, the average throughput $T_k(t)$ in~\eqref{eq_PF_scheduler_Tk_update} evolves very slowly, and the PF scheduling decision at time $t$ is dominated by the instantaneous throughput $R_k(t)$. So, the scheduler opportunistically selects the UE with the strongest instantaneous channel, which corresponds to the one which experiences a near-optimal RIS configuration. Moreover, for any finite RIS switching interval $T_s$, all RIS configurations are explored over time, ensuring that each UE eventually experiences a favorable channel realization and is thus scheduled for transmission. As a result, the system effectively achieves near-optimal RIS gains that is obtained under RR-scheduling, despite using only randomized RIS configurations. In essence, the PF mechanism of a $5$G system enables randomized RIS to procure optimal performance without explicit RIS optimization. 
\subsection{How Large a PF Scheduling Window Size, $T_c$ is Sufficient?}
Recall that Theorem~\ref{thm:infinite_Tc} applies in the asymptotic regime where $T_c \to \infty$, i.e., when the scheduling horizon is sufficiently long for each UE’s throughput to converge to its steady-state long-term average. However, a large $T_c$ comes at the cost of increased UE latency and implicitly requires an unrealistic assumption that the channels remain static over large time durations. 
Consequently, it is pertinent to characterize how $T_c$ should scale with system parameters in order for PF scheduling with randomized RIS to achieve near-optimal throughput. To this end, we note that $T_c$ effectively determines the time window over which the PF scheduler can exploit \emph{opportunism} by prioritizing UEs with favorable instantaneous channel conditions. Beyond this window, the fairness component of the scheduler begins to dominate the scheduling decisions. Then, motivated by the proof of Theorem~\ref{thm:infinite_Tc}, we identify the following two aspects as being critical in the choice of $T_c$:
\begin{enumerate}[leftmargin=*]
    \item \textbf{Convergence of RIS distribution:} To fully exploit opportunism across the UEs, the empirical distribution of RIS states over $\tilde{T}_c$ slots, denoted by $\hat{\mathbf{p}}\triangleq[\hat{p}_1,\ldots,\hat{p}_L]^T$, where $\hat{p}_\ell$ is the fraction of time RIS state $\phi_\ell$ is selected, must be close to the stationary RIS distribution $\mathbf{p}\triangleq[p_1,\ldots,p_L]^T$ in~\eqref{eq_final_sampling_distribution} with high probability.
     \item \textbf{Convergence of per-UE scheduling frequency:} In order for all UEs to realize their long-term throughput under PF scheduling, each UE must be scheduled sufficiently many times over the window of $\tilde{T_c}$ slots. To this end, let $S_k(t)\in\{0,1\}$ denote the scheduling indicator for UE-$k$ at time $t$, where $S_k(t)=1$ if and only if UE-$k$ is scheduled at time $t$. Define the empirical scheduling frequency vector
        \begin{equation}\label{eq_emp_sched_freq_vector}
            \hat{q}_k(\tilde{T_c}) \triangleq \frac{1}{\tilde{T_c}}\sum\nolimits_{t=1}^{\tilde{T_c}} S_k(t),\text{ and } q_k \triangleq \mathbb{E}[S_k(t)] = \frac{1}{K}, \ \forall k.
        \end{equation}
    Then, suffices to choose $\tilde{T_c}$ such that the empirical scheduling frequency vector $\hat{\mathbf{q}}(\tilde{T_c})\triangleq\left[\hat{q}_1(\tilde{T_c}),\ldots,\hat{q}_K(\tilde{T_c})\right]^T$ is close to $\mathbf{q}\triangleq\left[q_1,\ldots,q_K\right]^T$ with high probability.
\end{enumerate}
We can formally capture the above two requirements using the
\emph{total variation (TV) distance} between two probability measures
$P$, $Q$ over a measurable space $(\Omega,\mathcal{F})$, defined as
\begin{equation}\label{eq_TVD_defn}
\|P-Q\|_{\mathrm{TV}}
\triangleq
\sup\nolimits_{A \in \mathcal{F}} |P(A)-Q(A)|.
\end{equation}
If the sample space $\Omega$ contains only finitely many elements, say $\Omega=\{1,2,\ldots,n\}$, then the associated probability measure vectors $P$ and $Q$ can be represented by their PMFs: $\mathbf{p}=[p_1,\ldots,p_n]^T$, and $\mathbf{q}=[q_1,\ldots,q_n]^T$,
where $p_i=P(\{i\})$ and $q_i=Q(\{i\})$.
Then \eqref{eq_TVD_defn} simplifies to
$
\|P-Q\|_{\mathrm{TV}}
=
\frac{1}{2}\sum\nolimits_{i=1}^n |p_i-q_i|.
$
Using this, let $\varepsilon_1,\varepsilon_2 > 0$ be accuracy parameters such that
\begin{equation}\label{eq_TVD_reqd}
    \|\hat{\mathbf{p}}-\mathbf{p}\|_{\mathrm{TV}} \le \varepsilon_1, \qquad
\|\hat{\mathbf{q}}-\mathbf{q}\|_{\mathrm{TV}} \le \varepsilon_2.
\end{equation}
Then, Theorem \ref{thm:finite_Tc} below, we characterize how to choose $\tilde{T}_c$ such that~\eqref{eq_TVD_reqd} holds with probability at least $1-\eta_1$ and $1-\eta_2$, respectively, where $\eta_1, \eta_2 \in (0, 1/2)$ are confidence parameters. We also derive a corresponding throughput guarantee.
\begin{theorem}\label{thm:finite_Tc}
Let $\varepsilon_1 $ and $\varepsilon_2$ be the accuracy parameters as in~\eqref{eq_TVD_reqd}, and $\eta_1,\eta_2 \in (0,1/2)$ be the corresponding confidence parameters. Let the long-term throughput vector of the $K$ UEs under PF scheduling be given by $\tilde{\mathbf{T}} \triangleq [\tilde{T}_1,\ldots,\tilde{T}_K]^T$, where $\tilde{T}_k \triangleq T_k(\tilde{T}_c)$, and $\mathbf{T}^*$ denote the optimal throughput vector given by $\mathbf{T}^* =
(1/K)
\,[R_1^{\mathrm{opt}},\ldots,R_K^{\mathrm{opt}}]^T$.
Under the setting considered in Theorem~\ref{thm:infinite_Tc}, if $\tilde{T}_c$ satisfies
\begin{equation}\label{eq:thm_reqd_Tc}
\tilde{T}_c \ge \max \Bigg\{\tilde{T_s}\cdot \frac{1}{2\varepsilon_1^2}
\ln \left( \frac{2L}{\eta_1} \right), \frac{1}{2\varepsilon_2^2}
\ln \left( \frac{2K}{\eta_2} \right)
\!\Bigg\},
\end{equation}
then, with probability at least $1-\eta_1-\eta_2$, the following holds:
\begin{equation}\label{eq:error_metric_normalized}
{\|\tilde{\mathbf{T}}-\mathbf{T}^*\|_\infty}\big/
{\|\mathbf{T}^*\|_\infty}
\le 2K(\varepsilon_1+\varepsilon_2).
\end{equation}
\end{theorem}
\begin{proof} 
See Appendix~\ref{appendix_proof_thm2}.
\end{proof}
\begin{figure}
\centering
\begin{subfigure}{0.48\linewidth}
    \centering
    \includegraphics[width=1.05\linewidth]{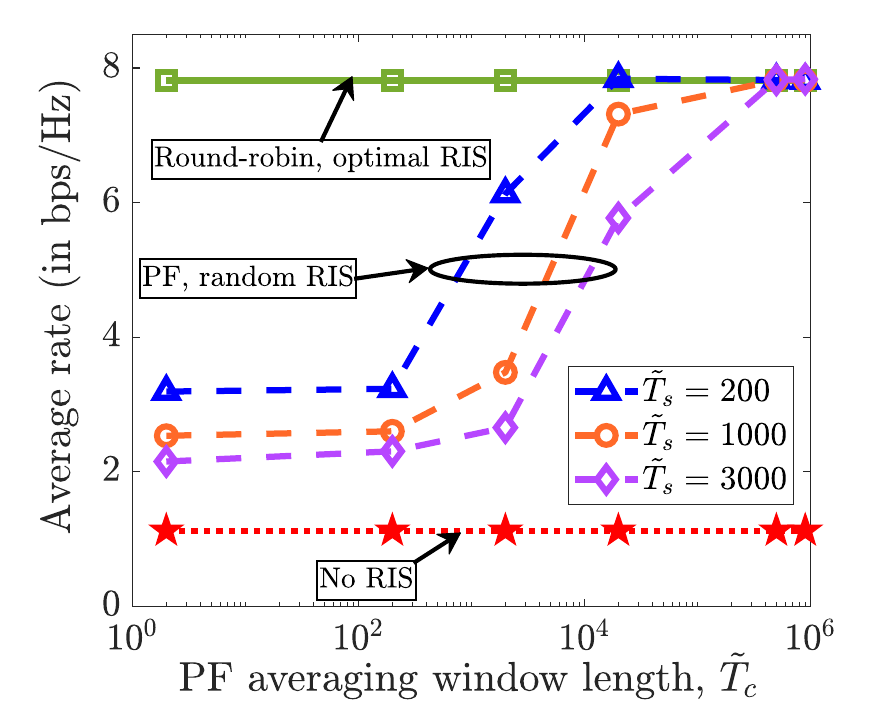}
    \caption{Throughput vs. $\tilde{T}_c$.}
    \label{fig:Thrm_1}
\end{subfigure}
\begin{subfigure}{0.48\linewidth}
    \centering
     \includegraphics[width =1.1\linewidth]{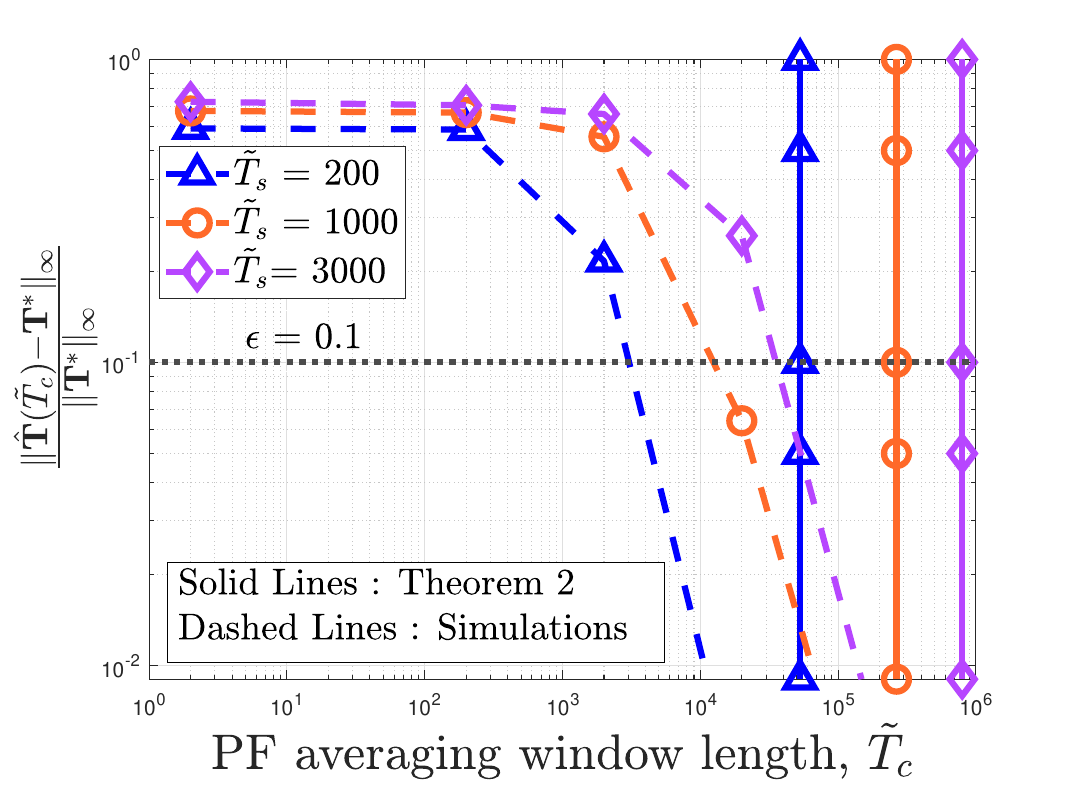}
    \caption{Optimality gap vs. $\tilde{T}_c$.}
    \label{fig:Thrm_2}
\end{subfigure}
\caption{Simulations validating the theory of RIS-assisted $5$G NR.}
\label{fig_simulations}
\end{figure}
From Theorem~\ref{thm:finite_Tc}, we deduce that, in order for the system to obtain most of the benefits of RIS, it is sufficient if $\tilde{T}_c$ scales only logarithmically in the number of RIS states $L$ and the number of UEs $K$, and this shows the favorable scalability of our scheme in practical wireless systems. Also, the linear dependence of $\tilde{T}_c$ on the RIS switching interval $T_s$ reflects the fact that larger values of $T_s$ result in fewer RIS state realizations within a given PF averaging window, thereby requiring a proportionally larger $\tilde{T}_c$ to ensure sufficient exploration of RIS configurations over its sample space.

\begin{remark}[Practical Design Implications]\label{remark:Choice_Ts}
Theorem~\ref{thm:finite_Tc} provides a quantitative guideline for selecting the PF averaging window $T_c$ relative to the RIS switching
interval $T_s$. Conversely, in practical deployments, a judicious choice of $T_s$ is also crucial: $T_s$ should be long enough to support $5$G NR DL control signaling, MCS adaptation, and data transmission to the UE currently benefiting from a near-beamforming RIS configuration, yet not so long that the PF scheduler (with a given EWMA parameter $T_c$) is compelled to serve other UEs to which the RIS may not deliver near-optimal gains. Thus, a judicious selection of the $(T_c,T_s)$ pair based on the $5$G NR scheduling and channel state feedback mechanisms in addition to following the bound in Theorem~\ref{thm:finite_Tc} is essential to fully realize the near-optimal benefits of randomized RIS. 
\end{remark}
\subsection{Simulation Results}
We now validate Theorems~\ref{thm:infinite_Tc} and~\ref{thm:finite_Tc}
via Monte Carlo simulations for an RIS-aided wireless system. A BS located at $(0,0)$~m communicates with $K=10$ UEs which are uniformly distributed over a $ [70,120] \times [-20,20]$~m grid  using an $N=256$-element RIS deployed at $(60,20)$~m. The path losses in all links follow a break-point model with reference distance path gain of $C_0=-30$~dB and path loss exponents $2.0$, $2.8$, and $3.8$ in the BS-RIS, RIS-UE, and BS-UE links, respectively \cite{Yashvanth_WCL_2025,3gpp_NR_path_loss}. The system operates at a transmit SNR of $110$~dB. The RIS employs $L=10$ phase states and the accuracy/confidence parameters as described in Theorem~\ref{thm:finite_Tc} are fixed to $\varepsilon=0.1$ and $\eta=0.1$.

In Fig.~\ref{fig:Thrm_1}, we plot the system throughput as a function of the PF averaging window length $\tilde{T}_c$ for different values of the RIS switching interval $\tilde{T}_s$, along with two baselines: (a) RR scheduling with optimal phase configuration for the selected user and (b) no-RIS. These two yield the highest and lowest throughput, respectively, and we note that (a) incurs significant training and feedback overhead compared to random RIS configuration, as one needs to find the optimal phase configuration for the user scheduled in each slot. 

For a given $\tilde{T}_s$, when $\tilde{T}_c$ is small, the throughput is well below RR with optimized phase configuration, since the PF scheduler is strongly constrained by prioritizing fairness and does not always transmit data to a UE to which the RIS provides its gain. However, notably, this throughput is still better than that achieved without an RIS, demonstrating that even a randomly configured RIS offers some gains to any user in its vicinity. As $\tilde{T}_c$ increases, the PF throughput improves by exploiting the opportunism in scheduling UEs. When $\tilde{T}_c$ is sufficiently large, the throughput converges to that of RR for any $\tilde{T}_s < \infty$, as expected from Theorem~\ref{thm:infinite_Tc}. 
 
Intuitively, this is because when $\tilde{T}_c$ is large and $\tilde{T}_s$ is finite, every RIS state is eventually realized over time, and by exploiting the opportunism of the PF scheduler, we can procure the near-optimal RIS benefits at every UE in the system. This validates our analytical result in Theorem~\ref{thm:infinite_Tc}. On the other hand, for a fixed and finite $\tilde{T}_c$, increasing $\tilde{T}_s$ reduces the throughput. A larger $\tilde{T}_s$ limits how quickly the RIS can explore its sample space within the EWMA window, thereby reducing its ability to exploit opportunism and multi-user diversity.

In Fig.~\ref{fig:Thrm_2}, we plot the throughput gap between randomly configuring the RIS and the RR benchmark, as a function of $\tilde{T}_c$. For any fixed $\tilde{T}_s$, the optimality gap decreases with $\tilde{T}_c$, in line with Fig.~\ref{fig:Thrm_1}. 
However, importantly, when $\tilde{T}_c$ is chosen according to Theorem~\ref{thm:finite_Tc}, the optimality gap lies well below the target threshold, $\varepsilon$, marked on the figure. This validates the sufficiency of the scaling of $\tilde{T}_c$ in Theorem~\ref{thm:finite_Tc}. 
\section{Randomized RIS-Assisted Opportunistic Scheduling in $5$G NR : Experiments}\label{Experiments}
\begin{figure}
    \centering
\includegraphics[width=\linewidth]
{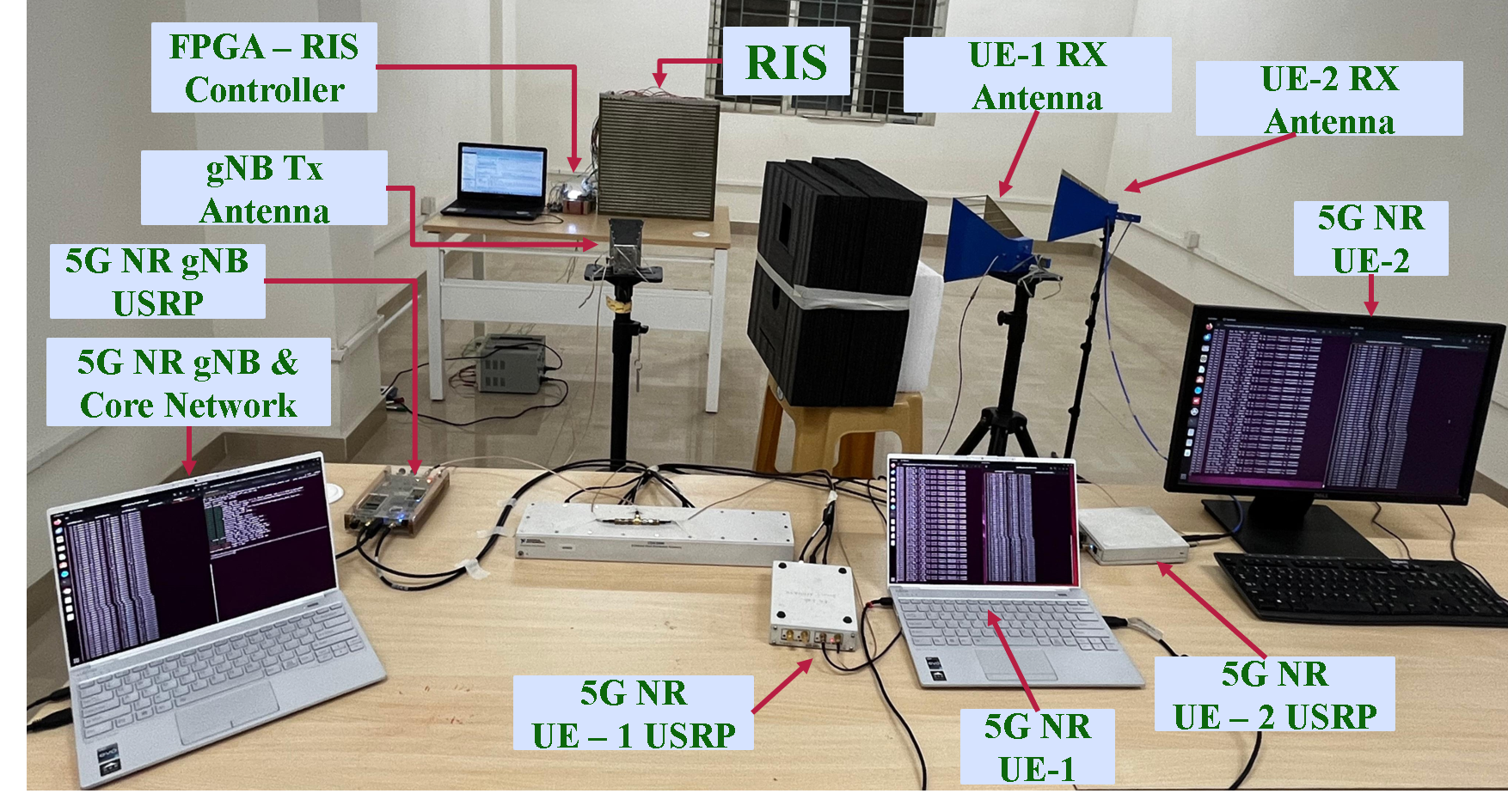}
\caption{Experimental setup used in this work. The system uses an OAI-based $5$G NR implementation with $1$ gNB and $2$ UEs connected to it via an RIS. A large block is used to obstruct the direct path from the gNB to the RIS.} 
    \label{fig:lab_setup}
\end{figure}
In this section, we present experimental results of a randomized RIS-assisted 5G NR system that implements real-time opportunistic scheduling of UEs connected to it. 
\subsection{Experimental Setup and Evaluation Methodology}\label{sec_data_collection}
Our experimental setup consists of a gNB and UEs; and the full protocol stack of 5G NR is implemented using OAI. USRP B$210$s serve as the radio units. 
External horn antennas were connected to the USRPs to provide additional directional gain to improve link reliability. The direct link between the BS and UEs is blocked to explicitly evaluate the impact of RIS on the performance. For simplicity and due to hardware availability, we consider $K=2$ UEs; the methodology extends naturally to more UEs also. The UEs are positioned such that $\nu_1 = 30^\circ$, $\nu_2 = 60^\circ$, and $\psi_1 = \psi_2 = 0^\circ$. The setup is shown in Fig.~\ref{fig:lab_setup}.

In line with Sec.~\ref{Sec_RIS_distribution}, the throughput-maximizing sampling distribution requires the RIS to operate across two phase-shift states~\cite[Theorem~1]{Yashvanth_WCL_2025} corresponding to the dominant angular directions of the two UEs. Specifically, the RIS phase configuration at time $t$ is sampled as
\begin{equation}\label{eq_sweeping_states}
\boldsymbol{\phi}(t) = 
\begin{cases}
\boldsymbol{\phi}_1 \triangleq \mathbf{a}(\nu_1)\otimes\mathbf{a}(\psi_1), & \text{with probability } \frac{1}{2},\\[4pt]
\boldsymbol{\phi}_2 \triangleq \mathbf{a}(\nu_2)\otimes\mathbf{a}(\psi_2), & \text{with probability } \frac{1}{2},
\end{cases}
\end{equation}
where $\mathbf{a}(\omega)$ is the array steering vector as defined in Sec.~\ref{Sec_RIS_distribution}. 
\begin{remark}
As the RIS states are chosen randomly at every switching interval, 
a UE attains optimal beamforming gain only when the selected RIS state aligns with its instantaneous channel. Further, a UE benefits from the RIS whenever the BS schedules it in a time slot where the RIS is pointing at it. The BS is oblivious of the RIS states, since the RIS is configured randomly and independently of the BS's operation. 
\end{remark}
For real-time transmission, the UEs request DL traffic using the transmission control protocol (TCP) over the $5$G NR PDSCH. Measurements are collected under two scenarios:
\begin{itemize}[leftmargin=*]
    \item \textbf{Single-UE scenario}: Only one UE is connected to the gNB at a time. For each UE, experiments are conducted both with and without the RIS. When the RIS is present, it is beamformed toward the corresponding UE under consideration. In all cases, time traces of the RSRP, BLER, MCS, and \texttt{iperf} throughput are recorded.

    \item \textbf{Scheduling scenario}: Both UEs are simultaneously connected and continuously request DL data from the gNB, which employs a PF scheduler to select the transmitting UE at each time slot. In this setup, the RIS switches \emph{randomly} between the two phase configurations defined in~\eqref{eq_sweeping_states}, with switching intervals of $T_s = 3$, $5$, and $15$~s. 
    For each experiment, the RSRP, BLER, MCS, and \texttt{iperf} throughput are recorded at both UEs with EWMA lengths: $T_c = 0.5, 1,$ and $10$ s. The default value for $T_c$ in OAI is $0.5$~s.
\end{itemize}
\subsection{Modifying the MCS Adaptation in OAI}\label{newmcs_adapt}
As discussed in Sec.~\ref{sec_MCS}, the OAI RAN performs MCS adaptation based on the channel state as inferred via the BLER. 
It initializes with a moderate MCS index and increases/decreases the MCS index by $1$ when the observed BLER is below/above $0.05$ and $0.15$, respectively. This results in very slow MCS updates as the BLER reacts to channel conditions only over several radio frames. It is not suitable for RIS-assisted systems, where channel conditions vary rapidly due to RIS switching. 

To address this, we modified the OAI codebase to enhance the default MCS adaptation by introducing a mechanism that uses UEs' reported RSRP to perform fast MCS adaptation. Since RSRP is consistently reported by the UE with a periodicity of $80$ ms, we first define a \emph{nominal MCS} based on the RSRP, denoted by $\text{MCS}_{\text{RSRP}}$, as follows. If the reported RSRP $> -102$~dBm, we set $\text{MCS}_{\text{RSRP}} = 25$; if the reported RSRP $< -115$~dBm, we set $\text{MCS}_{\text{RSRP}} = 10$; in all other cases we set we set $\text{MCS}_{\text{RSRP}} = 15$. We use these values because we find that the system typically does not support an MCS above $25$ at a BLER below $5\%$ even if the reported RSRP is above $-102$~dBm. Similarly, it does not support an MCS above $10$ at a BLER below $15\%$ if the reported RSRP is below $-115$~dBm. Thus, $\text{MCS}_{\text{RSRP}}$ forms a nominal maximum MCS around which the usual BLER-based MCS adaptation can be performed. During run-time, if the current MCS $-\  \text{MCS}_{\text{RSRP}} > 6$, we decrease the MCS by $6$; and if $\text{MCS}_{\text{RSRP}} \ -$ current MCS $> 6$, we  increase the MCS by $6$. In all other cases, we increase/decrease the MCS by $1$ depending on the observed BLER as explained above. This enables the BS to perform fast MCS switching, in line with the rapid channel variations induced by the RIS, in addition to the usual BLER-based MCS adaptation. In turn, it ensures that the scheduling decisions are better aligned with the real-time channel conditions and improves the throughput.
\subsection{Quantifying the Benefits of the RIS}
We begin by quantifying the benefits of deploying an RIS. We then analyze the temporal behavior of key performance metrics, including RSRP, BLER, MCS, and \texttt{iperf} throughput, under randomized RIS configurations and PF scheduling of UEs.
\begin{table}[t!]
\centering
\captionof{table}{RSRP and throughput performance with and without RIS.}
\begin{tabular}{|c|cc|cc|}
\hline
\multirow{2}{*}{\begin{tabular}[c]{@{}c@{}} \textbf{RIS} \\ \textbf{state} \end{tabular}} & \multicolumn{2}{c|}{\textbf{RSRP (in dBm)}}         & \multicolumn{2}{c|}{\textbf{\begin{tabular}[c]{@{}c@{}}Throughput \\ (in Mbps)\end{tabular}}} \\ \cline{2-5} 
& \multicolumn{1}{c|}{UE-$1$}    & UE-$2$    & \multicolumn{1}{c|}{UE-$1$}                              & UE-$2$                             \\ \hline
Without RIS                                                                                     & \multicolumn{1}{c|}{$-110$} & $-110$ & \multicolumn{1}{c|}{$25.2$}                             & $26.3$                            \\ \hline
With RIS                                                                                        & \multicolumn{1}{c|}{$-100$} & $-97$ & \multicolumn{1}{c|}{$51.1$}                             & $55.9$                            \\ \hline
\end{tabular}
\label{table_with_without_RIS}
\end{table}
Table~\ref{table_with_without_RIS} summarizes the average RSRP values observed at both UEs under the single-UE scenario, with and without the RIS. 
The RSRP enhances by $10$--$13$~dB, which corresponds to more than a  $10\times$ increase in received signal power, and this is attributed to the beamforming gain provided by the RIS. Moreover, since the RSRP directly impacts throughput, both UEs witness a clear boost in throughput when the RIS is present. In the absence of the RIS, the system operates at an average MCS index of $12$, and the corresponding data rate of about $25-26$~Mbps is close to the theoretical peak data rate of $26.7$~Mbps for that MCS. Similarly, in the presence of the RIS, the system operates at an average MCS index of $22$, and the data rate of about $53$~Mbps is in line with the theoretical peak data rate of $61.4$~Mbps.\footnote{See https://5g-tools.com/5g-nr-throughput-calculator/}
These experimental findings, obtained using a physical RIS integrated with an OTA $5$G NR testbed, demonstrate the practically achievable gains in RIS-aided communications.\footnote{The dataset collected during our experimental measurement is available here: \href{https://github.com/avreddy453/5G_RIS_Experiment_Data/tree/master}{https://github.com/avreddy453/5G\_RIS\_Experiment\_Data.git}}
\subsection{Performance of Randomized RIS with PF scheduling}
\begin{figure*}[hbt!]
\centering
\begin{subfigure}{.475\linewidth}
\centering
  \includegraphics[width=9cm]{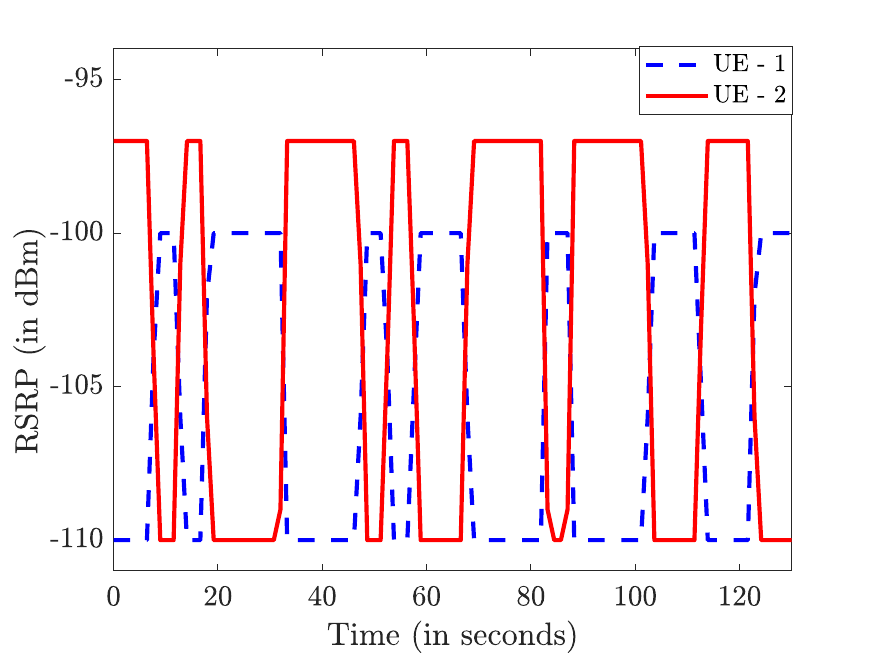}
  \caption{RSRP at the two UEs.}
  \label{fig_RSRP}
\end{subfigure}
\hspace{0.2cm}
\begin{subfigure}{.475\linewidth}
\centering
  \includegraphics[width=9cm]{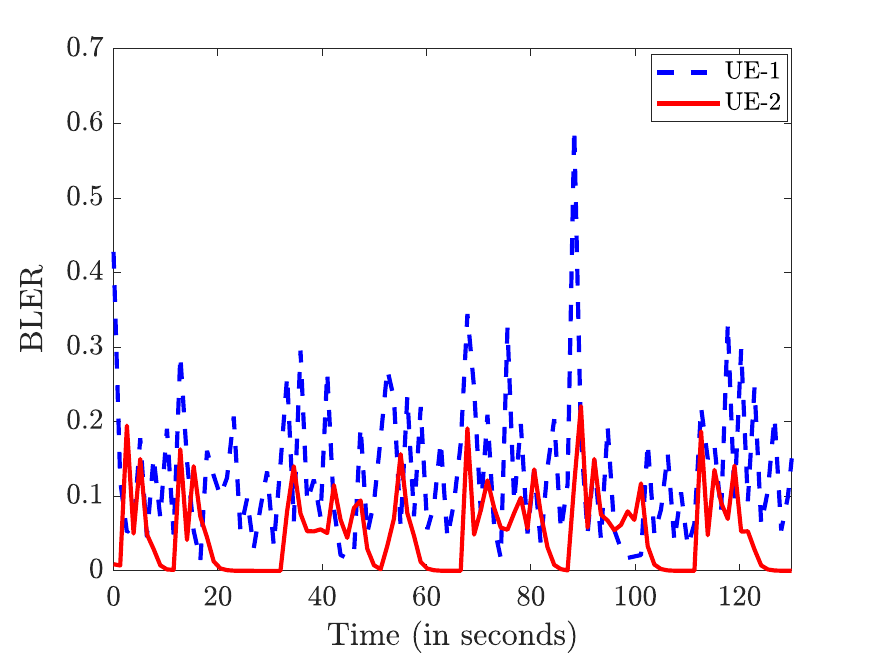}
  \caption{BLER at the two UEs.}
  \label{fig_BLER}
\end{subfigure}
\begin{subfigure}{.475\linewidth}
\centering
  \includegraphics[width=9cm]{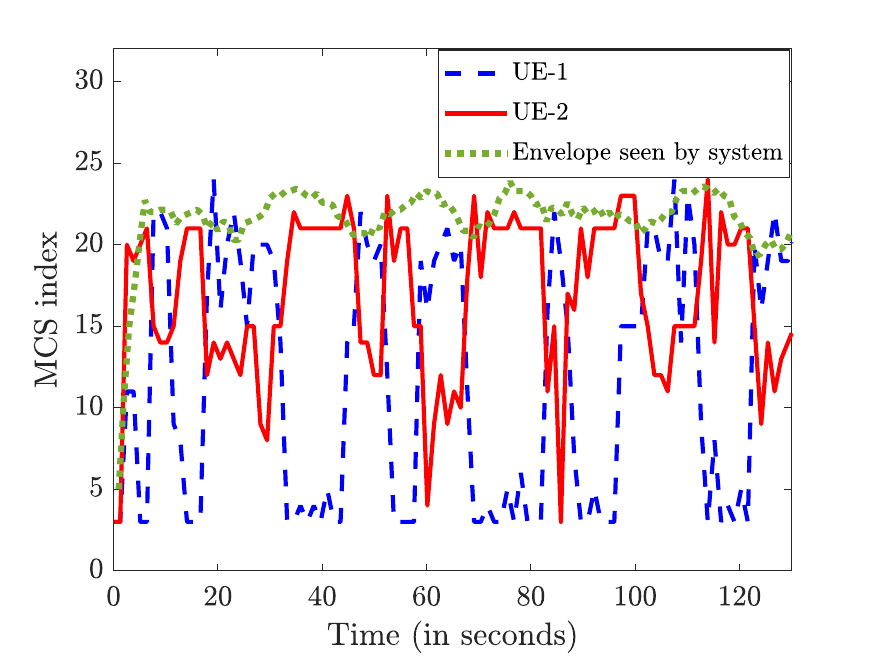}
  \caption{MCS indices for the two UEs.}
  \label{fig_MCS}
\end{subfigure}
\hspace{0.2cm}
\begin{subfigure}{.475\linewidth}
\centering
  \includegraphics[width=9cm]{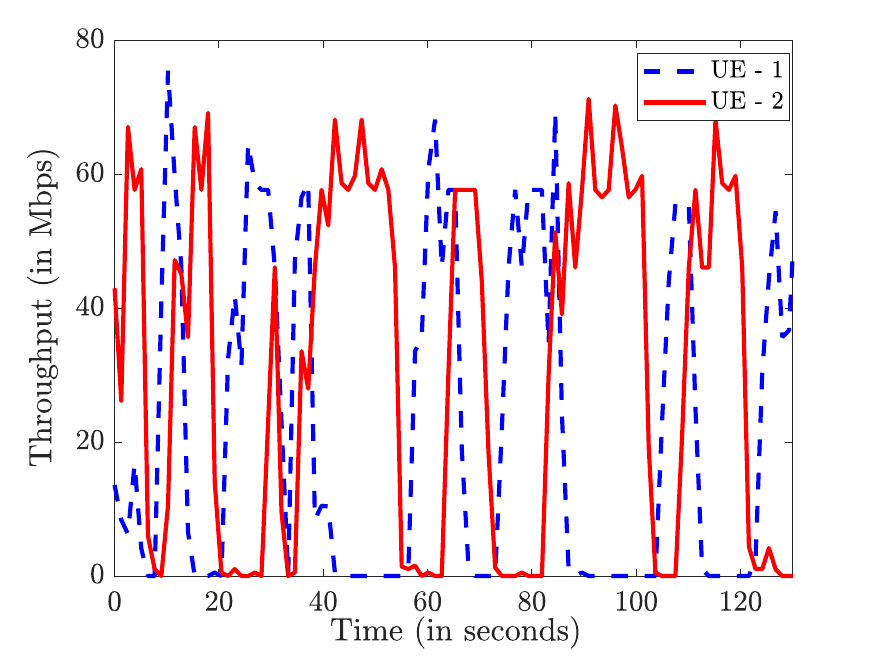}
  \caption{Throughput at the two UEs.}
  \label{fig_Trpt}
\end{subfigure}
\caption{Variation of performance metrics of a randomized RIS-aided $5$G NR system using a PF scheduler with $T_c=10$~s, and $T_s=5$~s.}
\label{fig:Time_trace}
\end{figure*}
We now examine the variation of performance metrics described in Sec.~\ref{sec_data_collection} under the scheduling scenario. Specifically, Fig.~\ref{fig:Time_trace} illustrates the temporal evolution of RSRP, BLER, MCS index, and throughput for both UEs when a PF scheduler with an averaging window of $T_c = 10$~s is employed. The RIS operates with a switching interval of $T_s = 5$~s, which was identified as the most effective configuration during experimentation. 
The rationale behind selecting this switching duration, along with the observed trade-offs associated with other system parameters, is discussed in Sec. \ref{tradeoff_ts}.
\subsubsection{Variation of RSRP}
Fig.~\ref{fig_RSRP} shows the variation of RSRP at both UEs over time. The RSRP for each UE alternates between high and low levels: $-100$~dBm and $-110$~dBm for UE-$1$, and $-97$~dBm and $-110$~dBm for UE-$2$. This behavior results from the RIS switching its phase configuration between the beamforming directions of the two UEs, as defined in~\eqref{eq_sweeping_states}. Since the RIS states are randomly chosen, the durations for which a UE experiences high or low RSRP vary, leading to different duty cycle lengths. When one UE experiences a high RSRP, the other UE has a low RSRP, and vice-versa.
\subsubsection{Variation of BLER}
In Fig.~\ref{fig_BLER}, we present the variation of BLER experienced by both UEs over time. The MCS adaptation based on BLER ensures that the BLER remains around the target value of $0.1$, but there are occasional excursions above/below the target value when the RIS switches away from/toward a scheduled UE, as expected. The performance of UE2 is generally better than that of UE1, possibly because of minor hardware differences in the setup. In addition, there are periods of time when the BLER of UE2 remains close to zero: ordinarily, the MCS adaptation should increase the MCS index to bring the BLER back up to $10\%$. However, these time periods are precisely the durations when the RIS is pointing towards UE1, causing UE1 to be predominantly scheduled. Since the MCS adaptation is based on calculating the BLER by averaging over about $100$ transport blocks, the MCS adaptation is slow during these time intervals.

\subsubsection{Variation of MCS Indices}
Fig.~\ref{fig_MCS} illustrates the temporal evolution of the MCS indices at the UEs. As discussed in Sec.~\ref{newmcs_adapt}, the selection of the MCS index during an RIS state transition at a given UE depends on the instantaneous channel quality (RSRP) between the gNB and that UE. From there on, the gNB dynamically varies the MCS based on the observed BLER. 
When the BLER increases beyond $0.15$, which typically occurs when the RSRP decreases from a high to a low level, the MCS index drops. This lowers the number of information bits transmitted per symbol, thereby improving decoding reliability and bringing the BLER back toward the target level. Similarly, when the BLER falls below $0.05$ (i.e., when the RSRP transitions from low to high), the gNB raises the MCS index, increasing the throughput while maintaining the desired reliability. During time intervals with no change in RSRP, the MCS index is relatively stable at a value that best matches the prevailing channel quality of the scheduled UE.

In Fig.~\ref{fig_MCS}, we also show the envelop of the MCS indices, which represents the effective MCS attained by the PF scheduler from a system perspective. Notably, even though the RIS configurations are selected randomly, the gNB consistently operates at relatively higher MCS levels. This highlights the effectiveness of the inherent opportunistic scheduling mechanism in $5$G, which prioritizes data transmission to a UE when it experiences favorable channel conditions.
\subsubsection{Variation of Throughput}
Fig.~\ref{fig_Trpt} presents the time evolution of the throughput achieved by the UEs at the application layer. Since a higher MCS index corresponds to transmitting more information bits per symbol, it directly translates to a higher throughput. Accordingly, consistent with Fig.~\ref{fig_MCS}, we observe that whenever the randomly configured RIS happens to direct its beam toward a given UE, that UE reports a high RSRP and experiences low BLER, resulting in the PF scheduler prioritizing data transmission to the UE and increasing its MCS, leading to a higher throughput.  
This is precisely the well-known \emph{multi-user diversity} effect, wherein the scheduler leverages the channel conditions across users to enhance overall performance. Consequently, as will be demonstrated in Fig.~\ref{fig_system_trpt}, this adaptive scheduling mechanism enables the system to achieve a substantially higher aggregate throughput, comparable to that obtained via optimized RIS that is tuned to point at the scheduled UE in every time slot.
\subsection{Randomized RIS Can Nearly Procure Full Benefits}\label{Sec_numerical_final}
In this section, we show that a randomized RIS, when paired with an opportunistic scheduler at the gNB, can achieve nearly the maximal throughput associated with optimized RIS configuration. Fig.~\ref{fig_system_trpt} depicts the EWMA system throughput (see \eqref{eq_PF_scheduler_Tk_update}) as a function of the PF window length $T_c$. 
The topmost curve corresponds to RR scheduling of UEs with the RIS optimally configured toward the scheduled UE in each slot. This genie-aided curve is obtained by averaging the single-UE throughput values with RIS beamforming to the considered UE (see Table~\ref{table_with_without_RIS}); it is independent of $T_c$ and represents the maximum achievable throughput with an optimized RIS. Practically, achieving this throughput requires tight control of the RIS, configuring it based on the scheduled UE. For comparison, we also include the case of RR scheduling with \emph{random} RIS switching. This curve appears below the PF scheduling curve, since RR does not exploit instantaneous channel variations and the consequent multi-user diversity. When the RIS switches randomly according to~\eqref{eq_sweeping_states} and the gNB employs PF scheduling, the system behavior is influenced by the averaging window $T_c$. For small $T_c$, the scheduler prioritizes short-term fairness, allocating resources more uniformly across users irrespective of the channel state and thereby reducing instantaneous throughput. As $T_c$ increases, the scheduler becomes more opportunistic and selects the UE that instantaneously benefits from the beamforming gain offered by the RIS. 
As a result, the PF throughput curve lies above the RR-with-random-RIS baseline and gradually approaches the genie-aided optimized RIS throughput as $T_c$ becomes large, confirming the opportunistic scheduling gain predicted in~\cite{Yashvanth_TSP_2023}. \emph{Importantly, the figure demonstrates that the throughput achieved with a randomized RIS and a properly tuned PF scheduler at the gNB not only consistently surpasses that of the no-RIS case (under both RR and PF scheduling), but can also deliver the full benefits of an optimized RIS even without explicitly optimizing the RIS.}

To further substantiate the above, Fig.~\ref{fig_histogram} shows a histogram of the fraction of time slots in which the PF scheduler selects a UE for transmission when the RIS is aligned with that UE’s channel and when it is not, based on a single long time trace. The scheduler predominantly transmits data to a UE when the RIS is in the beamforming configuration for that UE, but occasional instances occur where a UE is scheduled even when the RIS is not aligned toward it. This is because 5G NR prioritizes retransmission to a connected UE over new scheduling decisions when a NACK is received for a particular TB, resulting in the UE being scheduled even when it does not have the highest PF metric. Also, the UEs are scheduled almost equally over time, with each receiving about $40\%$ of the total time slots during which the RIS focuses its energy towards them, which depicts the long-term fairness in UE scheduling. 
\begin{figure*}
\centering
    \begin{subfigure}{0.49\linewidth}
    \centering
    \includegraphics[width=0.7\linewidth]{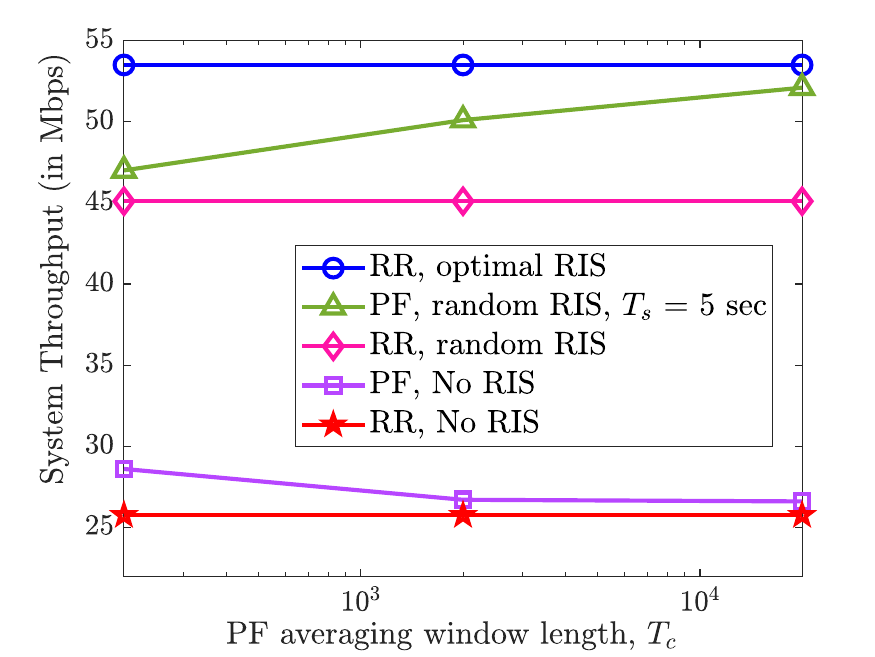}
    \caption{Throughput vs. $T_c$.}
    \label{fig_system_trpt}
     \end{subfigure}
    \begin{subfigure}{0.49\linewidth}
    \centering
    \includegraphics[width=0.9\linewidth]{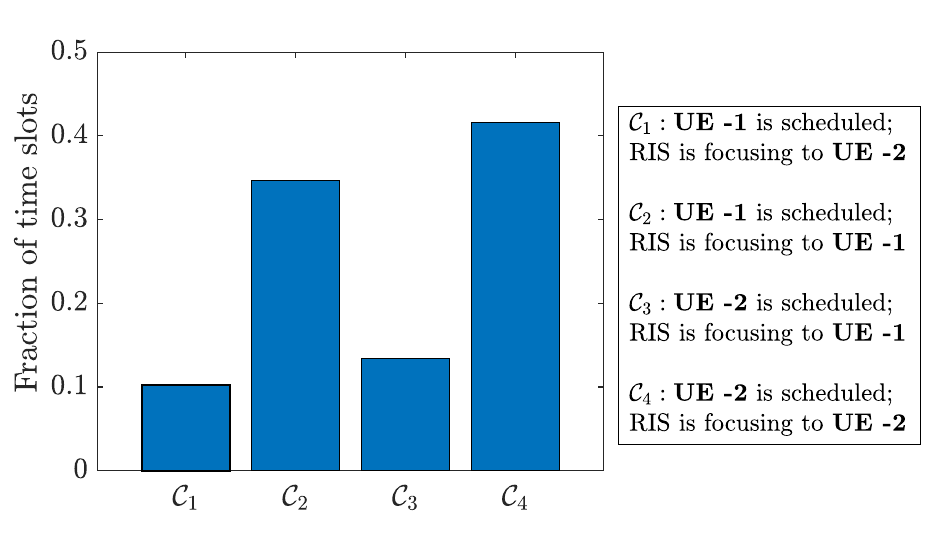}
    \caption{Fraction of time slots for four scheduling cases with $T_c = 10$ s and $T_s = 5$ s.} 
    \label{fig_histogram}
    \end{subfigure}
    \caption{Randomized RIS improves the throughput in real-time OTA $5$G NR systems.}
\label{fig_experimental_results}
\end{figure*}
\subsection{Tradeoff between $T_s$ and $T_c$}\label{tradeoff_ts}
In Table~\ref{table_Tc_Ts_tradeoff}, we report the throughput values obtained from experiments conducted with switching times $T_s \in \{1, 3, 5, 9, 15\}$~s and EWMA average window sizes $T_c \in \{0.1, 1, 10\}$~s. These are the same as the values used in the simulation setup of Fig.~\ref{fig_system_trpt}, but with the durations $T_s$ and $T_c$ measured in seconds. From Theorem~\ref{thm:finite_Tc}, at a fixed $T_c$, a smaller value of $T_s$ provides higher throughput, as it allows the RIS to explore more states within the throughput averaging window of $T_c$~s. On the other hand, if $T_s$ is very large,  the system fails to sufficiently explore independent RIS states within $T_c$ duration, and even a UE which does not witness favorable gain from the RIS gets scheduled for data transmission, limiting the achievable opportunistic gains. This is indeed the case in Table~\ref{table_Tc_Ts_tradeoff} as $T_s$ is reduced from $15$~s, for all three values of $T_c$. However, when $T_s$ becomes too small, there are overheads involved in MCS adaptation and other control steps, so the system has too little time to reach and operate at the optimal MCS corresponding to the RIS configuration.\footnote{Since the MCS is increased/decreased by $1$ based on the BLER computed using an averaging window of $100$~ms, setting $T_s$ to, say, $1$~s only allows for MCS adaptation $10$ times, which may not be sufficient to reach the optimal MCS before the RIS gets reconfigured to a different state.} 
Thus, an intermediate value of $T_s$ provides the best balance: it should be long enough to support optimal MCS selection under favorable RIS alignment at the scheduled UE, yet short enough to ensure that multiple RIS reconfigurations occur within $T_c$, thereby sufficiently exploiting the multi-user diversity. 
Hence, a switching interval of $T_s=5$ to $9$~s  achieves the best throughput in our setup. Furthermore, since the two users are scheduled with roughly equal frequency (see, e.g., Fig.~\ref{fig_histogram}), the scheme is also fair.

\begin{table}[t!]
\centering
\captionof{table}{Impact of EWMA averaging window $T_c$ and switching interval $T_s$ on the throughput.}
\label{table_Tc_Ts_tradeoff}
\resizebox{\columnwidth}{!}{%
\begin{tabular}{|c|c|c|c|c|c|}
\hline
\multirow{2}{*}{$T_c$ (s)} 
& \multicolumn{5}{c|}{\textbf{Throughput (Mbps)}} \\ \cline{2-6} 
& $T_s$ = 1 s & $T_s$ = 3 s & $T_s$ = 5 s & $T_s$ = 9 s & $T_s$ = 15 s \\ \hline
0.1 & 41.1 & 39.5 & 47.0 & 47.9 & 40.3 \\ \hline
1   & 41.0 & 41.6 & 50.1 & 48.2 & 42.9 \\ \hline
10  & 46.1 & 46.7 & 52.1 & 52.6 & 44.7 \\ \hline
\end{tabular}}
\end{table}
\subsection{Randomized RIS Helps Everywhere}
\begin{figure}
    \centering
    \includegraphics[width=\linewidth]{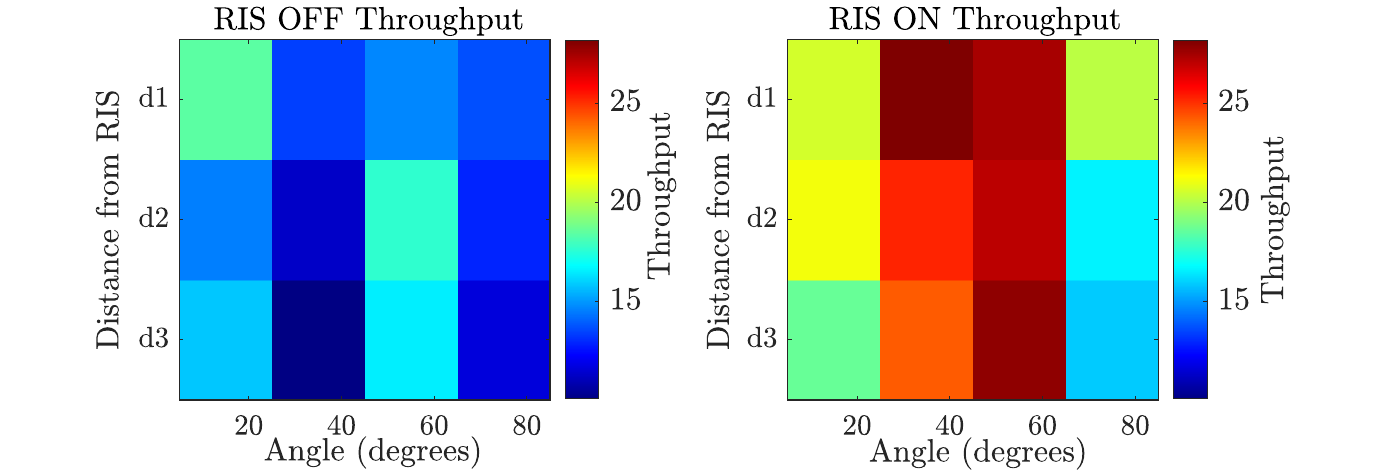}
    \caption{Throughput heatmap with and without RIS.}
    \label{fig_heatmap_randomRIS}
\end{figure}
In this final subsection, we show that even when RIS configurations are not randomly sampled according to beamforming-aligned directions for the UEs, the system still benefits from RIS states. We set $T_s = 5$~s and $T_c = 10$~s, and perform measurements for different UE angular locations and distances from the RIS under both RIS OFF and RIS ON operation. The experiments include the following scenarios: $(i)$ both UEs aligned with the RIS beamforming directions, $(ii)$ only one UE is aligned with an RIS beamforming direction, and $(iii)$ neither UE aligned with the RIS beamforming directions. 

For UE-1, the considered angular locations are $15^\circ$ and $30^\circ$, where $30^\circ$ corresponds to an RIS beamforming direction. Similarly, for UE-2, the considered angular locations are $75^\circ$ and $60^\circ$, where $60^\circ$ corresponds to an RIS beamforming direction. Accordingly, four UE angular placement configurations are considered:
$
00:(15^\circ,75^\circ), 
01:(15^\circ,60^\circ),
10:(30^\circ,75^\circ), 
11:(30^\circ,60^\circ).
$
For each angular configuration, the UE distances are varied over $d_1=1.22$~m, $d_2=1.52$~m, and $d_3=1.82$~m, while keeping the UE angles fixed. The RIS randomly switches between the $30^\circ$ and $60^\circ$ beamforming directions when ON. 

Fig.~\ref{fig_heatmap_randomRIS} shows the throughput for different UE angular locations and distances from the RIS under RIS OFF and RIS ON operation. The x-axis represents the angular location of an individual UE, while the y-axis represents the corresponding UE distance from the RIS. The throughput corresponding to a particular UE angle and distance is obtained by averaging over all configurations containing that angle. For example, for UE-1 located at $15^\circ$, the throughput values from the $00$ and $01$ configurations are averaged. 

When the RIS is OFF, comparatively lower throughput values are observed across most angle--distance combinations due to the absence of reflected beamforming gain. When the RIS is ON, the RIS randomly switches between beam directions centered around $30^\circ$ and $60^\circ$. Consequently, UEs located near these angular regions experience stronger reflected signals and achieve higher throughput values. Although the highest gains are observed near the RIS beamforming directions, noticeable improvements are also obtained for non-aligned UE locations due to the PF scheduler exploiting multi-user diversity by opportunistically serving the UE with favorable instantaneous channel conditions. Hence, randomized RIS operation provides throughput improvements throughout the considered coverage region, although the gains are smaller than those achieved by the configuration shown in Fig.~\ref{fig_system_trpt}.
\section{Conclusion}
In this work, we experimentally demonstrated the performance benefits of integrating a physically built RIS with an OTA real-time 5G system. We  showed that even randomly configured RIS can achieve near-optimal performance, both in theory and real-time experiments. By exploiting the inherent opportunistic scheduling mechanisms of 5G NR, the system naturally benefits from the channel variations induced by the RIS, without requiring explicit optimization of its configuration, thereby enabling a low-complexity deployment. A comprehensive analysis of RSRP, BLER, MCS, and throughput metrics further confirmed the effectiveness of randomized RIS configurations. In addition, an analytical study of the interaction between the RIS switching interval and the EWMA parameter used in the scheduler revealed its impact on the achievable system throughput. Our results further showed that a randomized RIS consistently provides performance gains compared to the baseline scenario without the RIS. Although the experiments were conducted in an indoor single-antenna setup with two users, the concept can readily extend to practical multi-antenna systems with a larger number of users. Evaluating randomized RIS performance in large-scale networks is part of our ongoing work.
\appendices
\section{Proof of Theorem 1}
\label{appendix_proof_thm1}
Recall that the angular domain is tiled into $L$ regions
$\{\mathcal{A}_\ell\}_{\ell=1}^L$, where each region is associated with a RIS phase configuration $\boldsymbol{\phi}_\ell$. Then,  
the set of UEs whose cascaded channel angles fall within region
$\mathcal{A}_\ell$ is given by
$
\mathcal{K}_\ell
=
\left\{
k \in \{1,\ldots,K\} :
(\nu_{c,k},\psi_{c,k}) \in \mathcal{A}_\ell
\right\}.
$
Accordingly, the fraction of UEs lying in region $\mathcal{A}_\ell$ is 
\begin{equation}
\psi_\ell(K)
=
{1}\big/{K}
\sum\nolimits_{k=1}^{K}
\mathbbm{1}\!\left\{
(\nu_{c,k},\psi_{c,k}) \in \mathcal{A}_\ell
\right\} \stackrel{(a)}{=} p_{\ell},
\end{equation}
where $(a)$ follows from~\eqref{eq_final_sampling_distribution}. 
So, the cardinality of $\mathcal{K}_\ell$ is $|\mathcal{K}_\ell| = p_\ell\, K$.
Now, suppose the PF scheduler operating as per
\eqref{eq:PF_metric}–\eqref{eq_PF_scheduler_Tk_update}, and that the index of the RIS state is $\tilde{\ell}$. Then, within class $\mathcal{K}_\ell$, the probability that UE-$k$ procures the maximum PF metric is 
\begin{equation}\label{eq:Qkl_def}
Q_{k,\ell}
\triangleq
\Pr\!\left(
{R_k^{\tilde{\ell}}(t)}\big/{T_k(t)}
>
{R_j^{\tilde{\ell}}(t)}\big/{T_j(t)},
\;
\forall j \in \mathcal{K}_\ell,\ j \neq k
\right),
\end{equation}
and this indicates the likelihood that UE-$k$ is scheduled among UEs in
$\mathcal{K}_\ell$ when the RIS is in configuration $\tilde{\ell}$. As a consequence, the average throughput achieved by UE-$k \in \mathcal{K}_\ell$
under PF scheduling, denoted by $T_{k,\mathrm{PF}}^{(K)}$, is given by
\begin{equation}\label{eq:pf_avg_throughput}
T_{k,\mathrm{PF}}^{(K)}
=
\beta_\ell(K)\,Q_{k,\ell}\,R_k^{\tilde{\ell}},
\end{equation}
where $\beta_\ell(K) \in [0,1]$ is the fraction of time allocated 
by the PF scheduler to UEs in class $\mathcal{K}_\ell$, which satisfies
$\sum_{\ell=1}^{L} \beta_\ell(K) = 1$. Also, since the instantaneous throughput at a UE under RIS vector index $\tilde{\ell}$ is upper bounded
by the optimal throughput in~\eqref{eq_opt_benchmark},
\begin{equation}\label{eq:pf_avg_throughput_ub}
T_{k,\mathrm{PF}}^{(K)}
\le T_{k,\mathrm{PF}}^{(K)*} \triangleq
\beta_\ell(K)\,Q_{k,\ell}\,R_k^{\mathrm{opt}}.
\end{equation}
To facilitate our analysis further, we introduce an auxiliary scheduler which operates in the following manner:
Whenever the RIS operates in configuration $\ell$, it
selects a UE in $\mathcal{K}_\ell$ with the largest PF metric, given by~\eqref{eq:Qkl_def}. 
To evaluate the throughput achieved under the auxiliary scheduler, recall that the RIS state changes once every $T_s$~s. Let $\tilde{T}_s$ as the number of slots within $T_s$~s. In this view, let $N_k$ be the number of slots in which UE-$k \in \mathcal{K}_\ell$ is scheduled. Then,
$N_k \sim \mathrm{Binomial}(\tilde{T}_s, Q_{k,\ell})$.
Hence, the average throughput achieved by UE-$k \in \mathcal{K}_\ell$ under the above auxiliary scheduler is $T_{k,\mathrm{AS}}^{(K)}=$
\begin{equation}\label{eq_trpt_aux}
\frac{p_\ell}{\tilde{T}_s}
\mathbb{E}\!\left[
\sum_{t=1}^{N_k} R_k^{\mathrm{opt}}
\right]
=
\frac{p_\ell}{\tilde{T}_s}\,
\mathbb{E}[N_k]\,
R_k^{\mathrm{opt}}
=
p_\ell\,Q_{k,\ell}\,R_k^{\mathrm{opt}}.
\end{equation}
We now relate the throughput achieved by the two scheduling schemes.  Recall that, when $T_c \to \infty$, the PF scheduling policy maximizes the following system utility:
\begin{equation}\label{eq_utility}
U_{\mathrm{PF}}(K) = \sum\nolimits_{k=1}^{K} \log T_k^{(K)},
\end{equation}
over all feasible schedulers~\cite{PF_max} and $T_k^{(K)}$ is the long-term average throughput of UE-$k$. Since the auxiliary scheduler is restricted to selecting a UE within $\mathcal{K}_\ell$, its utility is a lower bound on the utility achieved by the PF scheduler as $T_c \rightarrow \infty$. Then, using the optimality of the PF scheduler w.r.t.~\eqref{eq_utility}, as $T_c \rightarrow \infty$, we can relate~\eqref{eq:pf_avg_throughput} and~\eqref{eq_trpt_aux} as
\begin{equation}
\sum\nolimits_{k=1}^{K} \log T_{k,\mathrm{AS}}^{(K)}
\le
\sum\nolimits_{k=1}^{K} \log T_{k,\mathrm{PF}}^{(K)}
\le
\sum\nolimits_{k=1}^{K} \log T_{k,\mathrm{PF}}^{(K)*}.
\label{eq:sandwich}
\end{equation}
Now, by splitting the summation according to the $L$ classes, and upon substitution of the relevant expressions for all the terms and  further simplifications, the inequality between first and last terms in~\eqref{eq:sandwich} can be simplified as  
\begin{equation}\label{eq_KL_div}
\sum\nolimits_{\ell=1}^{L}
p_\ell
\log\!\left({p_\ell}\big/{\beta_\ell(K)}\right)
\le 0.
\end{equation}
We observe that the left hand side of~\eqref{eq_KL_div} represents the Kullback-Leibler (KL) divergence between two probability measures, which is always nonnegative and equals zero if and only if
$\beta_l(K)=p_l$ for all $\ell \in [L]$.
Hence, for any $K$, 
\begin{equation}\label{eq_KL_div_outcome}
\lim\nolimits_{T_c \rightarrow \infty} \beta_\ell(K) = p_\ell, \qquad \ell = 1,\ldots,L.
\end{equation}
Next, to evaluate $Q_{k,\ell}$, 
we note that UE--$k$ is selected by a PF scheduler if
$
{R_k^{\tilde{\ell}}(t)}\big/{T_k}
>
{R_j^{\tilde{\ell}}(t)}\big/{T_j},
\quad
\forall j\in\mathcal{K}_\ell,\ j\neq k.
$
Substituting \eqref{eq:Rk_inst} into the above condition and using the monotonicity of $\log(\cdot)$, the scheduling decision is characterized by
$
{
\log_2\left(
1+
\beta_k |\alpha_k|^2 C_\ell
\right)
}\big/{T_k}
>
{
\log_2\left(
1+
\beta_j |\alpha_j|^2 C_\ell
\right)
}\big/{T_j},
$
where
\(
C_\ell
\triangleq
\frac{P}{\sigma^2}
\bigl|
\boldsymbol{\phi}_{\tilde{\ell}}^{H}
\mathbf{a}_{\mathrm{RIS}}^{\mathrm{eff}}(\nu_{c,k},\psi_{c,k})
\bigr|^2
\)
is common to all UEs within $\mathcal{K}_\ell$.
Assuming identical large-scale fading for UEs belonging to the same class\footnote{For analytical tractability, we assume identical path loss for the UEs within a class \cite{Asymptotic_analysis}. This is used only to simplify the analysis; the core idea extends even with unequal path losses also, similar to~\cite{pathloss_diff}}, i.e., $\beta_k=\beta_j$ for all $k,j\in\mathcal{K}_\ell$, and invoking ergodicity, we can show $T_k=T_j$. Then the above simplifies to
$
|\alpha_k|^2 > |\alpha_j|^2,
\quad
\forall j\in\mathcal{K}_\ell,\ j\neq k.
$
Now, define $X_k \triangleq |\alpha_k|^2$ which is distributed as: $X_k \sim \mathrm{Exp}(1)$.
Then $Q_{k,\ell}$ can be computed as
$
Q_{k,\ell}
=
\Pr\!\left(
X_k > X_j,\ \forall j\in\mathcal{K}_\ell,\ j\neq k
\right).
$
To evaluate this value, let $K_\ell \triangleq |\mathcal{K}_\ell|$. Then, conditioning on $X_k(t)=x$, and averaging over $X_k(t)$ we obtain
\begin{align}
Q_{k,\ell}
&\stackrel{(b)}{=}
\int_{0}^{\infty}
\Pr\!\left(
X_j < x,\ \forall j\neq k
\right)
f_{X_k}(x)\,dx \label{eq:Qkl_conditional} \\
&\stackrel{(c)}{=}
\int_{0}^{\infty}
\left(1-e^{-x}\right)^{K_\ell-1}
e^{-x}\,dx, \label{eq:Qkl_integral}
\end{align}
where in $(b)$, $f_{X_k}(x)$ denotes the PDF of $X_k$ and in $(c)$, we used the expressions for the PDF and CDF of $X_k$. It is easy to show that the above integral evaluates to $Q_{k,\ell} = {1}\big/{K_\ell}$. Finally, applying the Sandwich Theorem to~\eqref{eq:sandwich}, and substituting the expressions derived in~\eqref{eq_KL_div_outcome} and for $Q_{k,\ell}$ in~\eqref{eq:pf_avg_throughput}, we obtain
\begin{equation}\label{eq_theorem_statement_per_UE}
T_{k,\mathrm{PF}}^{(K)} = {R_k^{\mathrm{opt}}}\big/{K}.
\end{equation}
Summing both sides of~\eqref{eq_theorem_statement_per_UE} over $k$, the proof follows.
\section{Proof of Theorem 2}
\label{appendix_proof_thm2}
The main idea of the first part of the proof is to characterize the confidence with which~\eqref{eq_TVD_reqd} holds. Within $\tilde{T}_c$ time slots corresponding to the PF averaging window, the RIS switches between $\tilde{T}_c/\tilde{T}_s$
possible configurations with state $\boldsymbol{\phi}_\ell$ being selected with probability $p_\ell$
(see~\eqref{eq_final_sampling_distribution}). 
Let $X_\ell$ be the number of times the RIS
configuration $\boldsymbol{\phi}_\ell$ is selected within $\tilde{T}_c$ slots; then we can show that $X_\ell \sim \mathrm{Binomial}(\tilde{T}_c/\tilde{T}_s,p_\ell)$. Further, the empirical PMF of the RIS state taking the value $\boldsymbol{\phi}_\ell$ is given by $\hat{p}_\ell = \frac{X_\ell}{\tilde{T}_c/\tilde{T}_s}$. Then, by
Hoeffding’s inequality~\cite{Hoeffding_Journal}, we can write
$
\Pr\!\left(|\hat{p}_\ell-p_\ell|>\varepsilon_1\right)
\le
2\exp\!\left(-2\frac{\tilde{T}_c}{\tilde{T}_s}\varepsilon_1^2\right).
$
Taking a union bound over all $L$ configuration states, we obtain
\begin{equation}\label{eq_hoeff_RIS_schedule}
\Pr\!\left(\exists \ell:\,|\hat{p}_\ell-p_\ell|>\varepsilon_1\right)
\le
2L\exp\!\left(-2({\tilde{T}_c}\big/{\tilde{T}_s})\varepsilon_1^2\right).
\end{equation}
Now, we can show that when the first lower-bound on $\tilde{T}_c$ given in~\eqref{eq:thm_reqd_Tc} holds, the probability bound in~\eqref{eq_hoeff_RIS_schedule} is at most~$\eta_1$, which proves the first term in~\eqref{eq:thm_reqd_Tc}. Similarly, from~\eqref{eq_emp_sched_freq_vector}, $\hat{q}_k = \frac{1}{\tilde{T_c}}\sum\nolimits_{t=1}^{\tilde{T_c}} S_k(t)$. Since $S_k(t)\in\{0,1\}$ across $t$ are bounded, the sequence $\{S_k(t)\}_{t=1}^{\tilde{T}_c}$ obeys the bounded difference property. Hence, using the Azuma-Hoeffding’s inequality, we get
\begin{equation}\label{eq_q_k_hat_deviation}
\Pr\!\left(|\hat{q}_k - q_k| > \varepsilon_2\right)
\le
2\exp\!\left(-2 \tilde{T_c}\varepsilon_2^2\right).
\end{equation}
To obtain a valid bound across all UEs, we apply a union bound over the events in~\eqref{eq_q_k_hat_deviation} across all UEs, yielding
\begin{equation}\label{eq_hoeff_UE_schedule}
\Pr\!\left(\exists k:\, |\hat{q}_k - q_k| > \varepsilon_2\right)
\le
2K \exp\!\left(-2 \tilde{T}_c \varepsilon_2^2\right).
\end{equation}
Now, whenever the second lower-bound on $\tilde{T}_c$ in~\eqref{eq:thm_reqd_Tc} holds, the bound in~\eqref{eq_hoeff_UE_schedule} is at most $\eta_2$. 

What remains to prove is~\eqref{eq:error_metric_normalized}.
To this end, when~\eqref{eq:thm_reqd_Tc} holds, by virtue of~\eqref{eq_TVD_reqd}, RIS vector
$\boldsymbol{\phi}_\ell$ is realized with probability
$\hat{p}_\ell \in (p_\ell -  \varepsilon_1, p_\ell +  \varepsilon_1)$. 
Let $q_{k,\ell}$ and $\hat q_{k,\ell}$ denote the true and empirical conditional probabilities of scheduling UE-$k$ when the RIS state is $\boldsymbol{\phi}_\ell$, respectively. Further, let $\ell^*(k)$ denote the index of the RIS phase configuration yielding the maximum achievable rate for UE-$k$, i.e., $R_k^{\mathrm{opt}} = R_{k,\ell^*(k)}$. By the law of total probability, the long-term and empirical scheduling frequencies of UE-$k$ can be expressed as
$
q_k = \sum_{\ell=1}^{L} p_\ell q_{k,\ell},$ and $\hat q_k = \sum_{\ell=1}^{L} \hat p_\ell \hat q_{k,\ell}.
$ Also, when~\eqref{eq:thm_reqd_Tc} holds, UE-$k$ is scheduled with probability
$\hat{q}_{k} \in (q_{k} - \varepsilon_2, q_{k} + \varepsilon_2)$.
Accordingly, the empirical and optimal throughputs of UE-$k$ are respectively
\begin{equation*}
\tilde{T}_k 
= \sum\nolimits_{\ell=1}^{L} \hat{p}_\ell \, \hat{q}_{k,\ell} \, R_{k,\ell}, T_k^* 
= \frac{1}{K} R_k^{\mathrm{opt}}
= \sum\nolimits_{\ell=1}^{L} p_\ell \, q_{k,\ell} \, R_{k,\ell}.
\end{equation*}
where the last equality follows from the asymptotic PF throughput characterization established in Theorem~\ref{thm:infinite_Tc}.
Hence, 
\begin{equation}\label{eq_diff_trpt_initial}
\tilde{T}_k - T_k^*
= \sum\nolimits_{\ell=1}^{L} 
\left( \hat{p}_\ell \hat{q}_{k,\ell} - p_\ell q_{k,\ell} \right) R_{k,\ell}.
\end{equation}
Adding and subtracting $p_\ell \hat{q}_{k,\ell}$, we can rewrite~\eqref{eq_diff_trpt_initial} as
\begin{equation}
\tilde{T}_k - T_k^*
= \sum\nolimits_{\ell=1}^{L} 
\Big[
(\hat{p}_\ell - p_\ell)\hat{q}_{k,\ell}
+ p_\ell(\hat{q}_{k,\ell} - q_{k,\ell})
\Big] R_{k,\ell}.
\end{equation}
Taking the magnitude and applying the triangle inequality to the right hand side, we obtain
\begin{equation*}
\left| \tilde{T}_k - T_k^* \right|
\le
\sum_{\ell=1}^{L} 
|\hat{p}_\ell - p_\ell| \, \hat{q}_{k,\ell} \, R_{k,\ell} +
\sum_{\ell=1}^{L} 
|\hat{q}_{k,\ell} - q_{k,\ell}| \, p_\ell \, R_{k,\ell}.
\end{equation*}
Now, note that $\hat{q}_{k,\ell} , p_\ell \le 1, R_{k,\ell} \leq R_k^{\mathrm{opt}}$, and $\sum\nolimits_{\ell=1}^{L} |\hat{p}_\ell - p_\ell| \le 2\varepsilon_1,$ and $\sum_{\ell=1}^{L} |\hat{q}_{k,\ell} - q_{k,\ell}| = |\hat{q}_{k} - q_{k}| \le 2\varepsilon_2$ by~\eqref{eq_TVD_reqd}. Applying these inequalities to the above, we get
\begin{equation}
\left| \tilde{T}_k - T_k^* \right|
\le
(2\varepsilon_1 + 2\varepsilon_2)\, R_k^{\mathrm{opt}}.
\end{equation}
Using the expression for $T_k^*$, it follows that
\begin{equation}
    {|\tilde{T}_k-T_k^*|}\big/{T_k^*}
\le
2K(\varepsilon_1+\varepsilon_2), \forall \ k.
\end{equation}
Finally, taking the maximum over all $k$, the proof follows.

\bibliographystyle{IEEEtran}
\bibliography{IEEEabrv,References}	

\end{document}